\documentclass[aip, cha, twocolumn,  showpacs, floatfix, letterpaper, 10pt]{revtex4-2}
\usepackage{graphicx}
\usepackage{dcolumn}
\usepackage{bm}

\usepackage{amsfonts}
\usepackage[utf8]{inputenc}
\usepackage[T1]{fontenc}
\usepackage{amsthm}
\usepackage{amsmath}
\usepackage{amssymb}

\usepackage{epsfig}
\usepackage[percent]{overpic}
\usepackage{pict2e}
\usepackage{psfrag}
\usepackage{verbatim}
\usepackage{color}
\usepackage{placeins}

\usepackage[all]{xy}
\usepackage{enumerate}

\usepackage{tikz}
\usetikzlibrary{matrix}

\usepackage{mathptmx}

\theoremstyle{plain}
\newtheorem{thm}{Theorem}[section]

\newtheorem{rem}[thm]{Remark}

\theoremstyle{definition}

\newtheorem{ex}[thm]{Example}

\newtheoremstyle{myremark}
{3pt}
{3pt}
{\small \rmfamily}
{5pt}
{\rmfamily}
{:}
{.5em}
{}

\theoremstyle{myremark}

\def\R{\mathbb{R}}

\def\Z{\mathbb{Z}}

\def\txtd{{\textnormal{d}}}
\def\txte{{\textnormal{e}}}
\def\txti{{\textnormal{i}}}

\def\ra{\rightarrow}
\def\I{\infty}

\newcommand{\be}{\begin{equation}}
\newcommand{\ee}{\end{equation}}
\newcommand{\benn}{\begin{equation*}}
\newcommand{\eenn}{\end{equation*}}
\newcommand{\bea}{\begin{eqnarray}}
\newcommand{\eea}{\end{eqnarray}}
\newcommand{\beann}{\begin{eqnarray*}}
	\newcommand{\eeann}{\end{eqnarray*}}

\newcommand{\myendex}{$\blacklozenge$\end{ex}}
\newcommand{\myendexerc}{$\lozenge$\end{exerc}}
\newcommand{\myendpexerc}{$\lozenge$\end{pexerc}}

\newcommand{\rd}{\mathrm{d}}

\newcommand{\ti}{\rightarrow\infty} 
\newcommand{\tz}{\rightarrow0}  
\newcommand{\Cst}{C^{\natural}} 
\newcommand{\rmean}{\bar{r}}  
\newcommand{\rmin}{r_\text{min}}
\newcommand{\rmax}{r_\text{max}}

\newcommand{\rmeanRZ}{\langle\rmean\rangle} 
\newcommand{\rminRZ}{\langle\rmin\rangle}
\newcommand{\rmaxRZ}{\langle\rmax\rangle}
\newcommand{\im}{\txti}
\newcommand{\Tt}{T_\text{tr}}
\newcommand{\evat}[1]{\biggr\rvert_{#1}} 

\begin{document}

\title{Graphop Mean-Field Limits and Synchronization for the Stochastic Kuramoto Model}

\author{Marios Antonios Gkogkas} 
\affiliation{Department of Mathematics, Technical University of Munich, 85748 Garching b.~M\"unchen, Germany}
\author{Benjamin J\"uttner}
\affiliation{Department of Applied Mathematics and Computer Science, Technical University of Denmark, 2800 Kgs. Lyngby, Denmark}
\author{Christian Kuehn}
\affiliation{Department of Mathematics, Technical University of Munich, 85748 Garching b.~M\"unchen, Germany}\affiliation{Complexity Science Hub Vienna, 1070 Vienna, Austria}
\author{Erik Andreas Martens}
\email[Corresponding author: ]{erik.martens@math.lth.se}
\affiliation{Centre for Mathematical Sciences, Lund University, S\"olvegatan 18, 221 00 Lund, Sweden}
\date{\today} 

\begin{abstract}
Models of coupled oscillator networks play an important role in describing collective synchronization dynamics
in biological and technological systems. The Kuramoto model describes oscillator’s phase evolution and explains
the transition from incoherent to coherent oscillations under simplifying assumptions including all-to-all coupling
with uniform strength. Real world networks, however, often display heterogeneous connectivity and coupling
weights that influence the critical threshold for this transition. We formulate a general mean field theory (Vlasov-Focker
Planck equation) for stochastic Kuramoto-type phase oscillator models, valid for coupling graphs/networks with
heterogeneous connectivity and coupling strengths, using graphop theory in the mean field limit. Considering
symmetric odd-valued coupling functions, we mathematically prove an exact formula for the critical threshold for the
incoherence-coherence transition. We numerically test the predicted threshold using large finite-size representations of
the network model. For a large class of graph models, we find that the numerical tests agree very well with the predicted
threshold obtained from mean field theory. However, the prediction is more difficult in practice for graph structures that are sufficiently
sparse. Our findings open future research avenues toward a deeper understanding of mean-field theories for heterogeneous systems. 
\end{abstract}

\maketitle

\begin{quotation}
{\bf
Networks of coupled oscillators appear in an impressive range of systems in nature and technology where they display collective dynamics, such as synchronization
~\cite{PikovskyBook2001,Strogatz2003,Buzsaki2006rhythms}.
The Kuramoto model describes the phase evolution of oscillators~\cite{Kuramoto1975,Kuramoto1984} and explains the transition from incoherent to coherent synchronized oscillations for a critical threshold of the coupling strength under simplifying assumptions, such as all-to-all coupling with uniform strength~\cite{Strogatz2000,Acebron2005}; however, real world networks often display strong heterogeneity in connectivity and coupling strength, which affect the critical threshold~\cite{Gleeson2012}. We derive a mean field theory for stochastic Kuramoto-type models and extend it to a large class of heterogeneous graph/network structures via graphop descriptions valid for the mean-field limit. We prove a mathematically exact formula for the critical threshold, which we test numerically for large finite-size representations of the network model.
}
\end{quotation}


\section{Introduction}
The discovery of synchronization dates back to 1665 with Christiaan Huygens' observations of two synchronizing pendulum clocks~\cite{Huygens1967}, and its mathematical modeling likely began with Norbert Wiener who was inspired by neuronal oscillations in the brain~\cite{Strogatz1994}. Wiener's formulation of the problem, however, was too general to allow for any analytical progress; simplifying assumptions were necessary to render the problem mathematically tractable~\cite{Winfree1967}, culminating in Yoshiki Kuramoto's paradigmatic model~\cite{Kuramoto1975,Kuramoto1984}. Kuramoto's original model describes the time evolution of the oscillator phases $\theta_k=\theta_k(t)$,
\begin{equation*}
\frac{\txtd}{\txtd t}\theta_k=:\dot{\theta}_k = \omega_k + \frac{C}{N}\sum_{j=1}^N  \sin(\theta_j - \theta_k),
\end{equation*}
where $k \in \{1,\ldots N\} =:[N]$, the coupling interaction between oscillators is first order, the coupling is all-to-all with uniform strength $C$, and the intrinsic frequencies $\omega_k$ are drawn unimodally from a distribution $g$ centered in the origin. The level of synchronization in this transition is aptly captured using the order parameter  $r(t)=\frac{1}{N}\left|\sum_{j=1}^N \exp{(\txti \theta_j(t))}\right|$, which tends to 0 when oscillator phases are incoherent (disordered) for weak coupling or to 1 when oscillators lock their frequencies and phases clump together (we say that the phases are coherent / the oscillators synchronize). When frequencies are identical, $\omega_k=\omega_j$ for all $k,j\in[N]$, the so-called synchronization manifold defined by $\theta_k(t)=\theta_j(t)$ for all $k,j\in[N]$ exists and is attractive for $C>0$; vice versa, when frequencies are non-identical (or symmetry is broken due to some other mechanism, see below), the loss or gain of coherence plays out in a competition between the strength of the heterogeneity and coupling strength $C$. Thus, for a set distribution width of the intrinsic frequencies, Kuramoto's model exhibits a transition from incoherent to coherent oscillations as the coupling strength $C$ surpasses a certain threshold value $\Cst$.

Kuramoto's initial heuristic analysis was based on a self-consistency equation for the order parameter~\cite{Sakaguchi1986}, allowing one to predict the critical coupling strength associated with the incoherence-coherence transition. A more formal mathematical treatment, facilitating deeper insights would, however, require a mean field theory valid in the limit, $N\rightarrow\infty$. Such a theory~\cite{Gupta2018} describes the dynamics in terms of a density function in the oscillator phases, $\rho=\rho(\phi,t)$, which evolves according to a transport equation (formally, a Vlasov-Fokker-Planck equation; see  Eq.~\eqref{eq:VFPE1}). Such a description was used by Strogatz and Mirollo~\cite{Strogatz1991} to investigate the stability of the incoherent branch for $C<\Cst$ where $\rho=1/(2\pi)$ by studying the associated eigenvalue spectrum and to (re-)derive the critical coupling $\Cst=2/(\pi g(0))$, where $g(0)$ denotes the maximum value of a unimodal frequency distribution $g$;
this approach was further developed and applied to variants of the Kuramoto model~\cite{Acebron1998}.
Other studies focused on the stability analysis for the partially synchronized branch ($C>\Cst$)~\cite{Mirollo2007}. An exact low dimensional description in terms of the macroscopic dynamics (order parameter) allowing to express the evolution of the order parameter in terms of an ordinary differential equation became available later ~\cite{Watanabe1993,Ott2008b,BickMartens2020}.

While Kuramoto's simplifying assumptions allowed for making significant progress in the mathematical understanding of the synchronization phenomenon, to understand real world oscillator dynamics,
it is desirable to break these assumptions toward increasing complexity. There are a number of ways of doing this; here, we are concerned with how the incoherence-coherence transition is affected by the presence of (thermal) noise and, in particular, network heterogeneities, which play a major role in real systems~\cite{GarciaOjalvo2012,Jeong2000,Brockmann2006,Castellano2009,Bullmore2009,Suki1994,Rohden2012,Kaluza2010}.
Indeed, the ability of coupled oscillators to synchronize has been investigated under the influence of noise~\cite{Sakaguchi1986,Son2010}, heterogeneous connectivity~\cite{Strogatz2001,Restrepo2006}, or heterogeneous coupling, such as non-local~\cite{Kuramoto2002,Panaggio2014}, $k$-nearest-neighbor~\cite{Wiley2006}, or random coupling strengths~\cite{SherringtonKirkpatrick1975,Ko2008} and also on experiments~\cite{Kiss2002,Taylor2009,Calugaru2020,MartensThutupalli2013} where oscillators are subject to real world influences.

Mean-field descriptions for $N\rightarrow \infty$ are well established for various theoretical frameworks including coupled oscillator networks~\cite{Gupta2018}. Our focus thus lies on mean-field limits valid for complex networks~\cite{Gleeson2012}, i.e., to generalize the Vlasov-Fokker-Planck (VFPE) equation (see Eq.~\eqref{eq:VFPE1}) so that it is capable of accurately describing the dynamics in complex networks characterized by heterogeneities in the connectivity or coupling strength. In particular, this includes cases where the adjacency matrix defining interactions between finitely many vertices is neither a full graph nor a highly symmetric structure, such as a lattice. In order to incorporate such structures, it is necessary to extend the description of (weighted) graph structures to the mean-field limit. This is possible via so-called graphons, which rely on concepts of the theory of limits of graph sequences~\cite{lovasz2006limits,lovasz2012large} or even more generally utilizing the theory of graphops~\cite{BS20}. Intuitively, graph limit theory provides a way to arrange limits of discrete graphs as continuous objects. Graphons achieve this, mostly within the context of dense graphs, using a coupling kernel function that describes the connectivity in the limit. Graphops generalize graphons, also incorporating many intermediate and sparse density graph limits in addition. Graphops can be represented as operators or via an associated measure-theoretic representation; i.e., they are generalizing purely kernel-based operators to more general operators. Recent studies have used graph limit theories to pursue the goal of heterogeneous mean-field limits. Several mathematical approaches have been successful in providing rigorous proofs for VFPEs, where nonlocal integral terms appear to take into account the heterogeneous coupling structure~\cite{ChibaMedvedev,KaliuzhnyiVerbovetskyiMedvedev1}. Recently, a general theoretical framework based on graphops has been put forward (by some of the authors of this paper) that allows us to generalize mean-field limit VFPEs easily from particular cases (nonlocal coupling or standard all-to-all) to describe modern complex network structures~\cite{Kuehn2020,GkogkasKuehn,KuehnXu,GkogkasKuehnXu1}.

In the present paper, we extend previous work~\cite{Kuehn2020,GkogkasKuehn,KuehnXu} to the stochastic case and formally derive a mean-field description based on graphop theory for the Kuramoto model with identical oscillators interacting via first order harmonic coupling function and non-uniform coupling strengths under the influence of (thermal) noise. We then derive rigorous results for the critical threshold for the incoherence-coherence transition ($\Cst$) by deriving a stability formula for the incoherent solution branch. A difficulty arises as it is unclear what demarcates the boundary of validity of mean-field PDEs for complex heterogeneous graphs; i.e., at some level of graph heterogeneity, it may be too difficult to accurately capture details of very sparse graph structures. As we cannot be sure under what circumstances our results correspond to the dynamics obtained for finite graphs (rigorous convergence results are still needed), we carry out detailed numerical simulations to test our results for various finite graph structures.

This article is structured as follows. In Sec.~\ref{sec:derivation}, we introduce a formal derivation of the mean-field equations for $N\rightarrow \infty$. In Sec.~\ref{sec: chap energy methods, linear stability of the incoherent state}, we derive the critical coupling strength $\Cst$ for the continuum limit, based on the graphop mean-field limit equation. In Sec.~\ref{sec:transition}, we carry out numerical simulations to investigate how the incoherence-coherence transition point predicted by the mean-field theory carries over to finite graphs for a range of graph structures, including dense and sparse topologies. Finally, we discuss our results in Sec.~\ref{sec:discuss}.


\section{Formal Derivation of the Mean-Field Equations}
\label{sec:derivation}
As discussed above, we are interested in mean-field models for stochastic Kuramoto(-type) models on networks~\cite{RO15}. The individual dynamics for the coupled identical oscillators is given by
\begin{equation}
\label{eq: stoch kuramotos model on graphs}
\frac{\txtd \theta^N_k}{\txtd t}=:\dot{\theta}_k^N = \frac{C}{N}\sum_{j=1}^N A^N_{kj} D(\theta_j^N - \theta_k^N) + \sqrt{2 \beta^{-1}} \dot{W}_k,\\
\end{equation}
where $k \in [N]:= \{1,...,N\}$ and $\theta_k=\theta_k(t) \in \mathbb{T}:=\R/(2\pi\Z)$ denotes the phase of the $k$th oscillator, $(A^N_{k,j})_{k,j \in [N]}$ denotes a weighted and non-negative adjacency matrix of the network (i.e., a graph $G$ with adjacency matrix $A^N_{kj}$), $D:\mathbb{T}\to \mathbb{R}$ is a sufficiently regular coupling function (e.g., $D=\sin$), $C>0$ is the coupling strength, $\beta>0$ is a diffusion constant controlling the noise level, and $W(t)=(W_1(t),\ldots,W_N(t))^\top$ is a vector of $N$ independent Brownian motions so that $\dot{W}_k=\dot{W}_k(t)$ is just a white noise forcing for each oscillator. The following derivation extends~\cite{Kuehn2020} to the stochastic case. To understand the formal derivation, let us consider the classical case of the Kuramoto model with all-to-all coupling with uniform strength, i.e., for $A^N_{k,j} = 1$ for all $k,j \in [N]$ and $D=\sin$. In other words, \eqref{eq: stoch kuramotos model on graphs} now reads as
\begin{equation}
\label{eq: stoch kuramotos model on graphs 1}
\dot{\theta}_k^N = \frac{C}{N}\sum_{j=1}^N  \sin(\theta_j^N - \theta_k^N) + \sqrt{2 \beta^{-1}} \dot{W}_k, \quad k \in [N].
\end{equation}
Let us introduce the complex order parameter
\begin{equation}\label{eq:orderparameter}
r\txte^{\txti\psi} := \frac{1}{N}\sum_{j=1}^{N}\txte^{\txti \theta_j}.
\end{equation}
Multiplying this equation by $\txte^{-\txti \theta_k}$ and equating imaginary parts, we have
\begin{equation}
r\sin(\psi - \theta_k) = \frac{1}{N}\sum_{j=1}^N  \sin(\theta_j^N - \theta_k^N)
\end{equation}
which implies that
\begin{equation}
\label{eq: stoch kuramotos model on graphs 2}
\dot{\theta}_k^N = C r\sin(\psi - \theta_k)  + \sqrt{2 \beta^{-1}} \dot{W}_k, \quad k \in [N].
\end{equation}
From \eqref{eq: stoch kuramotos model on graphs 2}, the mean-field character of the problem is visible as the $k$th oscillator just feels the averaged input from all other oscillators so one can think of a single typical oscillator and aim to analyze its dynamics. Let $\rho(t,\theta)~\rd \theta$ denote the fraction of oscillators with phase between $\theta$ and $\theta+ \rd \theta$ at time $t$; i.e., $\rho$ is a probability density. Assuming a law of large numbers in the limit $N\to \infty$, we formally get
\begin{equation}
 r\txte^{\txti\psi} = \frac{1}{N}\sum_{j=1}^{N}\txte^{\txti \theta_j} \to \int_0^{2\pi} \txte^{\txti \phi}\rho(t,\phi)~ \rd \phi.
\end{equation}
Now, using the same trick as above (i.e., multiplying both sides of the last equation by $\txte^{-\txti u}$ and taking imaginary parts), Eq.~\eqref{eq: stoch kuramotos model on graphs 2}
becomes in the limit $N\to \infty$
\begin{equation}
\label{eq: stoch kuramotos model on graphs 3}
\dot{u} = C \int_0^{2\pi} \sin( \phi -u)\rho(t,\phi) ~\rd \phi + \sqrt{2 \beta^{-1}} \dot{W}.
\end{equation}
Finally, the continuity equation, also called the Vlasov-Fokker-Planck equation (VFPE), for the probability density $\rho$, respectively, for the law of the limiting process $u$, reads as
\begin{align}\label{eq:VFPE1}
\begin{split}
\partial_t\rho &= -\partial_\theta \Big( \rho V(\rho)
 \Big) + \frac{1}{\beta} \partial^2_\theta\rho,\\
 V(\rho)&:= C \int_0^{2\pi} \sin(\phi - \theta)\rho(t,\phi)~\rd \phi.
 \end{split}
\end{align}
In summary,~\eqref{eq:VFPE1} is a partial differential equation with a first-order transport/advection-type term with a nonlocal convolution term involving the sine-nonlinearity mediating the coupling and with a second-order spatial diffusion term arising directly from the white noise forcing.

Now, let us come back to Eq.~\eqref{eq: stoch kuramotos model on graphs}. In this case, the next natural generalization step is to assume that the network (i.e., a graph $G$ with adjacency matrix $A^N_{kj}$) is sufficiently connected and does not have components, which are more connected than others; see also Ref.~\cite{RO15}. Moreover, we assume that
there exists a local order parameter $r_k \txte^{\txti \psi_k}$, which is locally proportional to a single global order parameter $r\txte^{\txti\psi}$ weighted by the degree $\kappa_k$ for each node; i.e., we have
\begin{equation}
\kappa_k r \txte^{\txti \psi} = r_k \txte^{\txti \psi_k} := \sum_{j=1}^N A^N_{kj}\txte^{\txti\theta_j}.
\end{equation}
By multiplying the local order parameter by $\txte^{-\txti\theta_k}$ and equating the imaginary parts in the last equation, we obtain
\begin{equation}
\label{eq: stoch kuramotos model on graphs 4}
\dot{\theta}_k^N = C \, r\, \kappa_k\sin(\psi - \theta_k)  + \sqrt{2 \beta^{-1}} \dot{W}_k, \quad k \in [N].
\end{equation}
Now, let $\rho(t,\theta,\kappa)~\rd \theta$ denote the probability for the fraction of oscillators having a phase between $\theta$ and $\theta+ \rd \theta$ and a degree $\kappa$ at time $t$. Note carefully that we have added an additional variable $\kappa$ to the density, which captures the (degree) heterogeneity of the network. If we assume that the network is uncorrelated and has degree distribution $d(\kappa)$, one is tempted to assume that in the limit $N\to \infty$, we have
\begin{equation}
 r\txte^{\txti\psi} = \frac{1}{\kappa_k}\sum_{j=1}^N A^N_{kj} \, \txte^{\txti\theta_j}
 \to
 \int_0^{2\pi}\int_0^\infty \txte^{\txti\phi} \frac{\kappa \, d(\kappa)}{\langle \kappa \rangle } \rho(t,\phi,\kappa)  ~\rd \kappa~\rd \phi,
\end{equation}
where $\langle \kappa \rangle$ is the average degree of a vertex in the graph and  $\frac{\kappa d(\kappa)}{\langle \kappa \rangle }\rho(t,\phi,\kappa)$ is the probability density for an edge having its end at a vertex of phase $\phi$ and degree $\kappa$ at time $t$.
Now, using the same trick as before, Eq.~\eqref{eq: stoch kuramotos model on graphs 4} becomes in the limit $N\to \infty$
\begin{equation}
\label{eq: stoch kuramotos model on graphs 5}
\dot{u} = \frac{C}{\langle \kappa \rangle}\int_0^{2\pi} \int_0^\infty \sin(\phi - u)~ l ~d(l) ~\rho(t,\phi,l) ~\rd\phi~ \rd l  + \sqrt{2 \beta^{-1}} \dot{W}.
\end{equation}
The continuity equation for the probability density $\rho$, respectively, the law of the limiting process $u$, reads as
\begin{subequations}
\begin{align}\label{eq: VFPE2}
\partial_t\rho &= -\partial_\theta \Big( \rho V[G](\rho)
\Big) + \frac{1}{\beta} \partial_\theta^2\rho,\\
V[G](\rho) &:= \frac{C}{\langle \kappa \rangle}\int_0^{2\pi} \int_0^\infty \sin(\phi - \theta)~ l~ d(l)~\rho(t,\phi,l)~ \rd \phi~ \rd l.
\end{align}
\end{subequations}
Thus, in comparison with the VFPE for the classical case of all-to-all coupling with uniform strength \eqref{eq:VFPE1}, we had to replace
\begin{equation}
\rho(t,\phi) \text{\quad by \quad}  \int_0^\infty \frac{ l d(l)}{\langle l \rangle }\rho(t,\phi,l) ~ \rd l.
\end{equation}
We can view this step as incorporating the structure of graph/network $G$ appearing in the Vlasov equation via an operator, which acts on the density $\rho$. In fact, one can even hope to completely remove averaging over the variable $\kappa$ that we used to capture the heterogeneity and just keep $\kappa$ as a new variable in the density, which then yields a whole hierarchy of mean-field VFPEs, one for each degree. This set of ideas can then be thought even further and one can directly replace the adjacency matrix by a coupling kernel and there are numerous papers in this direction~\cite{ChibaMedvedev,KaliuzhnyiVerbovetskyiMedvedev1,GkogkasKuehn,KuehnXu}.
Yet, it seems best to think of generalizing VFPEs more abstractly~\cite{Kuehn2020} by viewing the underlying network influence as given by some linear operator $A$ acting on the density so that a more abstract form of VFPEs would be given by
\begin{subequations}\label{eq:abstract_VFPE}
\begin{align}\label{eq: VFPE3}
\partial_t\rho &= -\partial_\theta \Big( \rho V[A](\rho)
\Big) + \frac{1}{\beta} \partial_\theta^2\rho,\\
V[A](\rho) &= C\int_0^{2\pi} D(\phi - \theta)  (A\rho)(t,\phi,x) ~\rd\phi,
\end{align}
\end{subequations}
where $x$ is a suitable variable that tracks the heterogeneity of the network so that one effectively obtains a family of VFPEs, and we have also replaced the sine-coupling again by a more general coupling function $D$. A typical choice of $x$ found in the literature would be to take it as a variable in the unit interval $x\in[0,1]=\Omega$, where points in the interval represent node labels in the infinite network limit ~\cite{ChibaMedvedev,KaliuzhnyiVerbovetskyiMedvedev1,GkogkasKuehn,KuehnXu}. Probably the most elegant abstract way to think of $A$ is as a graph operator, or graphop, as introduced in Ref.~\cite{BS20}. A graphop is a bounded, self-adjoint, and positivity-preserving operator $A: L^\infty(\Omega;m) \to L^1(\Omega; m)$, where $m$ is the reference measure on $\Omega$; e.g., one can pick the Lebesgue measure. To a given graphop $A$ always corresponds a family of finite measures $(\nu_x)_{x\in \Omega}$, called fiber measures, via the formula
\begin{equation*}
(Af)(x) = \int_{\Omega} f(y) ~\txtd\nu_x(y) \text{ \quad $x \in \Omega$\quad  for $f \in L^\infty(\Omega;m)$.}
\end{equation*}
Intuitively, we may view a graphop $A$ just as a generalized adjacency matrix for a symmetric graph and for a given node $x \in \Omega$, the fiber measure $\nu_x$ is just the edge distribution for this node. Indeed, for the finite-dimensional case, we can just pick $\Omega=[N]$ and $m$ as the uniform measure on $\Omega$, so that functions $f\in L^\infty(\Omega;m)$ can be identified with vectors in $\R^N$ and $Af$ is just the usual matrix-vector multiplication. Yet, we stress that in the limit $N\ra \I$, we need a space, such as $\Omega=[0,1]$, with the Lebesgue measure.

One may wonder, how far such an abstract construction for VFPEs involving graphops can work? It is clear that it works in simple cases, e.g., when the graph is all-to-all coupled as one can just drop the dependence on $x$. Also, if the graph is very dense and very regular with just two types of typical nodes, then one could take $x$ as a binary variable and so on. Furthermore, it is understood that it works for dense graphs, where $A$ can be represented by an integral operator with a sufficiently regular kernel, i.e., in the framework of so-called graphons. However, one does expect that there are growing networks as $N\ra \I$ that are so sparse and/or so heterogeneous that eventually, mean-field calculations may fail. Proving a precise boundary location on the space of networks to determine, when VFPEs are helpful and when they fail, seems out of reach at this point. Here, we take a pragmatic approach and start from the formal VFPE~\eqref{eq: VFPE3}, carry out stability analysis of the main bifurcation/phase transition to synchronization in the Kuramoto model, and then numerically simulate the dynamics for different discretized (i.e., finite-dimensional, large $N$) classes of graphops $A$ to check when the mean-field stability calculation is accurate. This is going to provide an indirect cross-check, whether a mean-field limit can work.


\section{Bifurcation/Phase Transition}
\label{sec: chap energy methods, linear stability of the incoherent state}
In the following, we consider~\eqref{eq: VFPE3}, and we assume for simplicity that
\begin{itemize}
        \item[(H0)] The coupling is non-trivial; i.e., $C>0$.
	\item[(H1)] $D$ is an odd $2\pi$-periodic function.
	\item[(H2)] $A$ is a graphop with a bounded $2\to 2$ norm; i.e., the following quantity exists and is finite:
	\begin{equation*}
	\parallel A \parallel_{2\to 2} := \sup_{v \in L^2(\Omega)} \frac{\parallel Av \parallel_2}{\parallel v \parallel_2} < \infty.
	\end{equation*}
\end{itemize}
This implies that $A$ can be uniquely extended to the Hilbert space $L^2(\Omega,m)$ (see Remark 2.12 in Ref.~\onlinecite{BS20}, for instance). For simplicity, we use the same notation for this extension; i.e., we write $A:L^2(\Omega,m) \to L^2(\Omega,m)$. For the solution $\rho(t,\theta,x)$ of \eqref{eq: VFPE3}, we define the $j$th Fourier coefficient as
\begin{equation}
z_j = \frac{1}{2\pi}\int_0^{2\pi} \txte^{- \txti j \theta}\rho(t,\theta,x)~\rd \theta, \quad j \in \mathbb{Z},
\end{equation}
where $\txti:=\sqrt{-1}$.
Note that we have effectively defined a family of Fourier coefficients that depends upon $x$, i.e., $\{(z_j)_x\}_{x\in\Omega}$, but we shall always write just $z_j$ in the calculation below and later discuss the $x$-dependence. Applying the Fourier transform to \eqref{eq: VFPE3}, exchanging integrals, and using integration-by-parts (in the second line), we have
\begin{widetext}
\begin{align}
\begin{split}
\label{ch enrgy methods, eq: Fourier trafo applied to VFPE 1}
\partial_t z_j &= \frac{1}{2\pi}\int_{\mathbb{T}}\txte^{- \txti j\theta} \Big(-\partial_\theta\{ \rho(t,\theta,x) V[A](\rho)(t,\theta,x)\} + \frac{1}{\beta} \partial^2_\theta \rho(t,\theta,x) \Big)~\rd \theta
\\
&=\frac{1}{2\pi} \left( \txti jC \int_\mathbb{T} \txte^{- \txti j\theta} \rho(t,\theta,x) \int_\mathbb{T} (A\rho)(t,\phi,x) D(\phi-\theta)~\rd\phi ~\rd \theta
- \frac{ j^2}{\beta} \underbrace{\int_\mathbb{T} \txte^{- \txti j\phi}\rho(t,\phi,x)~\rd\phi}_{= z_j}\right) \\
&=\frac{1}{2\pi} \left(\txti j C \sum_{l \in \mathbb{Z}}\hat{D}(l) \int_\mathbb{T} \txte^{ \txti(-j - l)\theta} \rho(t,\theta,x) \int_\mathbb{T} (A\rho)(t,\phi,x) \txte^{ \txti l \phi}~\rd\phi ~\rd \theta
- \frac{ j^2}{\beta} z_j\right)\\
&=\frac{1}{2\pi}  \Big(\txti j C \sum_{l \in \mathbb{Z}\setminus\{0\}}\hat{D}(l) z_{-j-l} Az_{-l} - \frac{ j^2}{\beta} z_j \Big),
\end{split}
\end{align}
\end{widetext}
where $\hat{D}$ denotes the Fourier transform of $D$, and in the last line, we used that $\hat{D}(0) = 0$, which follows from the fact that $D$ is an odd, periodic function. Moreover, $z_{-j} = \overline{z_{j}}$ holds, which follows from the fact that $\rho$ is real-valued. We can assume without loss of generality that $j \in \mathbb{N}$ to get the following system (i.e., the amplitude equation):
\begin{widetext}
    \begin{align}\label{eq:chap_energy_fourier_ana_for_stability_system_1}
    \begin{split}
    \partial_t z_j &= \frac{1}{2\pi} \Big\{ \Big( \txti C \hat{D} (-j) A - \frac{ j^2}{\beta} \Big)z_j +  \txti j C \sum_{l \in \mathbb{Z}, l \neq 0, -j}\hat{D}(l) z_{-j-l} Az_{-l} \Big\}, \quad
    j=1,2,\ldots.
    \end{split}
    \end{align}
\end{widetext}
The completely incoherent state $\rho_\I\equiv 1/(2\pi)$ of the oscillators corresponds to a uniform probability density over the circle, which translates into $z_0=1/(2\pi)$ and $z_j=0$ for all $j\neq 0$, and the state is also assumed to be independent of $x$; i.e., we assume that all different types of nodes are uniformly distributed across the circle for $\rho_\I$. Linearizing \eqref{eq:chap_energy_fourier_ana_for_stability_system_1} around this incoherent state yields via a straightforward calculation the system
\begin{subequations}
\label{eq, chap energy, fourie ana for stability, system 2}
	\begin{alignat}{4}
	\partial_t Z_j &= \frac{1}{2\pi} \Big\{\Big(  \txti C \hat{D} (-j) A - \frac{j^2}{\beta} \Big)Z_j \Big\} , \quad j = 1,2,\ldots,
	\end{alignat}
\end{subequations}
where we use $Z_j$ to denote the Fourier coefficients of the linearized dynamical system, and we observe that the linearized system nicely decouples. The question then is how does the stability of the $j$th Fourier mode depend on the eigenvalues of graphop $A$? On the Hilbert space $H:= L^2(\Omega,m)$, for any $j \in \mathbb{N}$, let us define the linearized operator $T_C^j: H \to H$,
\begin{equation*}
T_C^j w:=\frac{1}{2\pi}  \Big( \txti C \hat{D} (-j) A - \frac{j^2}{\beta} \Big) w.
\end{equation*}
Recall that the resolvent set $\rho(A)$ of the operator $A: H \to H$ is defined to be the set
\begin{align*}
\rho(A) := \{ \lambda \in \mathbb{C}: R_\lambda(A):= (A - \lambda I )^{-1}: \\
H \to H \text{ exists and is bounded} \},
\end{align*}
where $R_\lambda(A)$ is called the resolvent operator of $A$
and the spectrum of $A$ is the complement $\sigma(A) := \mathbb{C} \setminus \rho(A)$. Observe that for any $\lambda \in \mathbb{C}$, setting $\tilde{\lambda} := \frac{1}{2\pi} (\txti C \hat{D} (-j) \lambda  - \frac{ j^2}{\beta} ), $ we have
\begin{equation*}
R_{\tilde{\lambda}}(T^j_C) = \frac{1}{2\pi}   \txti C \hat{D}(-j) R_\lambda(A).
\end{equation*}
From this, we see that for all $j$ for which $\hat{D}(-j) \neq 0$, the condition that $R_\lambda(A)$ exists and is bounded is equivalent to the condition that $R_{\tilde{\lambda}}(T_C^j)$  exists and is bounded. From this, we conclude that for all $j \in \mathbb{Z}$ for which $\hat{D} (-j) \neq 0$, we have
\begin{equation*}
\sigma(T_C^j) = \frac{1}{2\pi} \Big( \txti C \hat{D} (-j) \sigma(A) - \frac{ j^2}{\beta} \Big).
\end{equation*}
For all other $j\in \mathbb{Z}$ (that is, for all $j$ for which $\hat{D} (-j) = 0$) we see immediately that $\sigma(T_C^j) = - \frac{ j^2}{\beta 2\pi}$. Since $A$ is bounded and self-adjoint, we have that $\sigma(A) \subset \mathbb{R}$ is a bounded set. Further note that since $D$ is an odd function, we must have $\hat{D}(j) = \txti \int_0^{2\pi} D(u) \sin( ju) \rd u \in \txti \mathbb{R}$. Finally, define
\begin{equation}
C^\natural := \inf \Big\{\frac{ j^2}{ \beta \txti \hat{D}(j) \lambda }: \lambda \in \sigma(A),\ j \in \mathbb{Z}^*,\ \txti\hat{D}(j) \lambda \geq 0 \Big\},
\label{eq: critical coupling}
\end{equation}
where $\mathbb{Z}^*:=\mathbb{Z}^+\cup \{0\}$. The next theorem shows that $C^\natural$ is a uniform parameter bound on the coupling strength independent of $x$, which means that smaller coupling leads to stability of incoherence, while above $C^\natural$, at least some classes of nodes synchronize at least partially.
More precisely, we have

\begin{thm} \textbf{(Incoherence-coherence transition)}
\label{prop: cha enrgy methods, incoherence-coherence transition}\\
Consider an odd, $2\pi$-periodic, continuous function $D: [0,2\pi] \to \mathbb{R}$ and a graphop $A: L^2(\Omega,m) \to L^2(\Omega,m)$. Then, the incoherent state $\rho_\infty$ is locally asymptotically stable for $0 < C < C^\natural$ and unstable for $C> C^\natural$.
\end{thm}

\begin{proof}
	Observe that for any $\lambda\in \sigma(A)$ and $j \in \mathbb{Z}^*$, such that $\txti \hat{D}(-j) \lambda<0$, the corresponding  element in the spectrum of $T_C^j$, $\tilde{\lambda}(C,j) = \frac{1}{2\pi}( \txti C \hat{D} (-j) \lambda - \frac{ j^2}{\beta}) \in \sigma(T_C^j)$, is strictly negative for any $C >0$; thus, it never crosses the imaginary axis. Thus, a crossing, for growing $C$, can occur only among those $\lambda\in \sigma(A)$ and $j \in \mathbb{Z}^*$ for which $\txti \hat{D}(-j) \lambda > 0$. Among all such $\lambda$ and $j$, the crossing occurs always at
	\begin{equation}\label{eq:crossing_occurs}
	C_{j ,\lambda}:= \frac{j^2}{ \beta \txti \hat{D}(-j) \lambda }.
	\end{equation}
	Observing that $C^\natural$ is the minimum of all these transition points, it follows immediately that $C^\natural$ is the smallest $C>0$ for which there exists $j \in \mathbb{N}$ such that
	an element in the spectrum of $T_C^j$ crosses the imaginary axis, namely, the element $\tilde{\lambda}(C^\natural,j) \in \sigma(T_{C^\natural}^j)$.
\end{proof}

As a first step, we want to carry out some specializations to examples and analytically consider some cases.

\begin{ex} \textbf{(Kuramoto model with first order interaction)}
\label{ex: crit couplin for class Kur}\\
We wish to specialize the general formula for the critical threshold in Eq.~\eqref{eq: critical coupling}, valid for Eq.~\eqref{eq:abstract_VFPE} [the continuum limit version of Eq.~\eqref{eq: stoch kuramotos model on graphs}], to the continuum limit version of the Kuramoto model in Eq.~\eqref{eq: stoch kuramotos model on graphs} such that oscillators interact via a first-order coupling function, i.e., 
$D(u)=\sin{u}$. Hence,
\begin{align*}
    \hat{D}(1) &= \frac{1}{2\pi} \int_0^{2\pi} \sin(u) \txte^{ \txti u} ~\rd u
    = \frac{\txti}{2}, \\
    \hat{D}(-1) &=  -\frac{\txti}{2}, \\
    \hat{D}(k) &= 0, \quad k \in \mathbb{Z}\setminus \{1,-1\}
\end{align*}
and $\txti \hat{D}(-1)  = \frac{1}{2} > 0$.
Then, by Theorem \ref{prop: cha enrgy methods, incoherence-coherence transition}, the incoherent state loses stability at
\begin{equation}\label{eq:Cst_classicalKM}
    C^\natural = \frac{2}{\beta\Lambda(A)},\,\,\,\,\Lambda(A):=\sup_{\lambda\in\sigma(A)}|\lambda|.
\end{equation}
\end{ex}

\begin{ex} \textbf{(Classical Kuramoto model on full graph)} \label{ex 2}\\
In the case of the \emph{full graph} (i.e., complete graph with uniform coupling strength), we have
\begin{equation*}
Af(x) = \int_\Omega f(y)~\txtd m(y), \quad x \in \Omega, f\in L^2(\Omega,m).
\end{equation*}
Clearly, $A$ is a non-invertible operator, and the only eigenvalue of $A$ is $1$ (The eigenvalue equation $Af = \lambda f$ implies that $f$ must be a constant, say $f_0\neq0$, satisfying $f_0 = \lambda f_0$. Thus, $\lambda=1$.) Moreover, in the case that $\lambda\in \mathbb{C} \setminus \{0,1\}$, the operator $A - \lambda I$  is invertible
since for any $g \in L^2(\Omega,m)$, the pre-image $f$ is achieved under the unique choice
\begin{equation*}
f := \frac{ \frac{c}{1 - \lambda} - g}{\lambda}, \quad c:= \int_\Omega g(y)~\rd m(y).
\end{equation*}
Thus, we have $\sigma(A) = \{ 0,1\}$.  Hence, for the Kuramoto model on the full graph, we obtain by Example \ref{ex: crit couplin for class Kur} that
\begin{equation}\label{eq:Cst_classicalKM_uniform}
C^\natural = \frac{2}{\beta}.
\end{equation}
\end{ex}
in agreement with previous analysis (see also, e.g., Ref.~\onlinecite{Strogatz2000}).

\begin{rem}
Sakaguchi~\cite{Sakaguchi} obtained for the critical coupling of the full graph the formula (in Sakaguchi's notation)
\begin{align}
K_C(D) &= 2 \Big( \int_{-\infty}^{\infty} \frac{1}{ \omega^2 + 1 }g(D\omega + \omega_0) ~\txtd\omega \Big)^{-1}.
\end{align}
In our framework, matching the assumptions and the notation correctly, we have $\omega_0= 0$, $D = \frac{1}{\beta}$,  and $g = \delta_0$. Note that in Sakaguchi's framework, the variance of the Brownian term $f_i(t)$ is $2Dt$, while in our framework, the Brownian term $\sqrt{2\beta^{-1}}W_k$ has variance $\frac{2}{\beta} t$ for each $k$; thus, we must have  $D = \frac{1}{\beta}$. Thus, Sakaguchi's formula simplifies to
\begin{align}
\begin{split}
K_C(D) &= 2 \Big( \int_{-\infty}^{\infty} \frac{1}{ \omega^2 + 1 }g(D\omega) ~\rd\omega\Big)^{-1}
\\
&= 2 \Big( \frac{1}{D}\int_{-\infty}^{\infty} \frac{1}{ \left(\frac{x}{D}\right)^2 + 1 } ~\rd\delta_0(x)\Big)^{-1} \\
&= 2 D = \frac{2}{\beta} = C^\natural,
\end{split}
\end{align}
which is exactly just the special case of the far more general formula we calculated in Example \ref{ex 2}.
\end{rem}

Although we have now a very nice formula for $C^\natural$, it is not immediately clear for which classes of networks this formula works as $N\ra \I$. After all, Theorem~\ref{prop: cha enrgy methods, incoherence-coherence transition} only makes claims about stability/instability based upon the assumption of the validity of the mean-field VFPE. Only if we already knew that the mean-field limit VFPE would be valid for certain classes of networks, i.e., if it does approximate --- in a suitable sense --- the oscillator system for finite but large $N$, then we could be certain applying our result for finite large networks. Proving such an approximation result in full generality is difficult, although first steps exist for the deterministic Vlasov case ~\cite{ChibaMedvedev,KaliuzhnyiVerbovetskyiMedvedev1,GkogkasKuehn,KuehnXu}. For example, one issue in this context is that the mean-field only holds in a scaling limit upon re-normalizing the sums appearing in the Kuramoto model suitably via the density of the graph. However, empirically testing the formula for $C^\natural$ via various classes of large finite networks using numerical simulation is certainly possible, and we shall proceed with this approach.


\section{Incoherence-coherence transition for finite and infinite oscillator networks}
\label{sec:transition}
\begin{figure}
    \begin{overpic}[width=\columnwidth,percent]{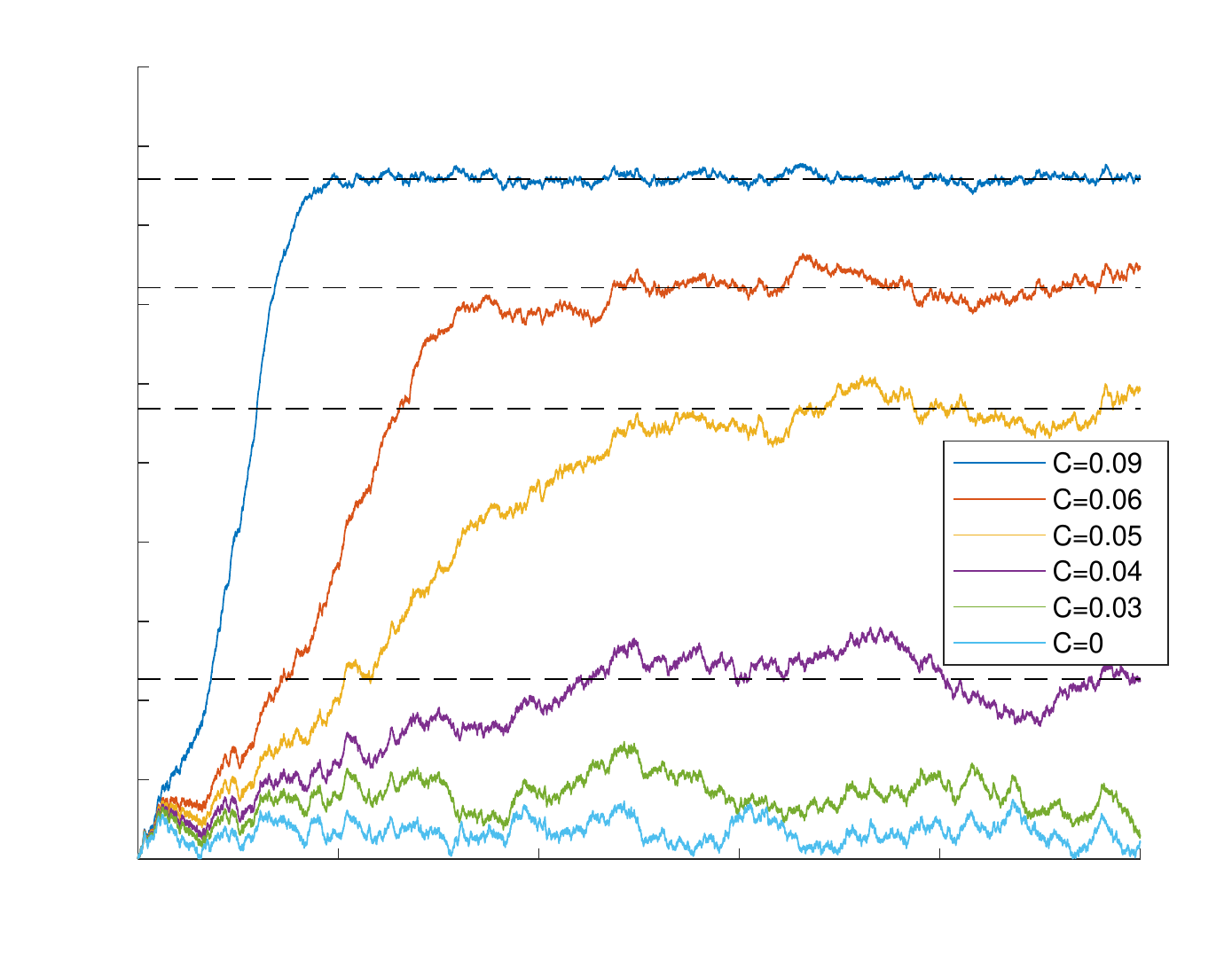}
        \put(8,71){\color{black}1}
        \put(8,8){\color{black}0}
        \put(4,40){\color{black}$r$}
        \put(87,4){\color{black}1000}
        \put(11,4){\color{black}0}
        \put(53,1){\color{black}$t$}
    \end{overpic}
    \caption{Time evolution of $r$ for numerical solutions of \eqref{eq: stoch kuramotos model on graphs 1} for different values of $C$.
    Clearly, the time traces of $r$ are subject to (random) fluctuations. Also, for higher $C$, the time traces of $r$ settle, after a transient, at some quasi-stationary state (dashed lines).
    Other parameters are $\beta=50$ and $N=1000$.}
    \label{fig:belowmidaboveCc}
\end{figure}

We want to check the prediction for the incoherence-coherence transition given in Theorem~\ref{prop: cha enrgy methods, incoherence-coherence transition} for the mean-field limit by numerical simulations. The challenges we face in doing so stem from the fact that numerical simulations are bound to a finite-dimensional representation of Eqs.~\eqref{eq: stoch kuramotos model on graphs} and to a finite simulation time. Thus, while Theorem~\ref{prop: cha enrgy methods, incoherence-coherence transition} can only hold in an  approximative sense for $N<\infty$, the finite system size and simulation time also incur uncertainty in the detection of the incoherence-coherence transition. Several points need to be taken into account when detecting the transition from incoherence to coherence that we outline below.

To see this, it is instructive to observe the dynamics for the case of the Kuramoto model where oscillators interact with $D(u)=\sin(u)$ on a complete graph with uniform coupling.
The collective dynamics of all oscillators is described by the order parameter $r(t)$ defined in~\eqref{eq:orderparameter} and is shown in Fig.~\ref{fig:belowmidaboveCc} for numerical solutions of \eqref{eq: stoch kuramotos model on graphs 1} for varying coupling strengths $C$ and fixed system size $N=1000$ and noise level $\beta=50$.
Initial phases are chosen to correspond to incoherent oscillations (see Sec.~\ref{numericalmethods} on numerical methods).
The dynamics of the order parameter $r$ is subject to fluctuations, which stems from two sources:
i) the stochastic dynamics inherent to the system and
ii) finite-size effects induce pseudo-random fluctuations
of order $\mathcal{O}(N^{-1/2})$ that vanish in the limit $N\rightarrow\infty$
\footnote{Finite size fluctuations are pseudo-random: finitely many oscillators move around the unit circle with distinct velocity differences and thus perpetually change their relative locations on the unit circle)}. After a transient time, $\Tt$, we observe that the dynamics settle into a quasi-stationary state (on average); i.e., the order parameter fluctuates around a constant mean value and is bounded by minimal and maximal values. If the trajectory after the transient attains a minimal value arbitrarily close to 0 during the observed time interval, we say that the population oscillates incoherently; if the minimal value never approaches 0, the dynamics are said to be (partially) coherent or synchronized (perfect synchrony occurs only for $r=1$), and we observe increasing synchrony for larger $C$.
Accordingly, Fig.~\ref{fig:belowmidaboveCc} allows us to distinguish incoherent oscillations for weak coupling strengths ($C=0$ to $C=0.03$), and partially coherent oscillations occur for stronger coupling ($C\geq 0.04$), which agrees well with the prediction of $ \Cst= 0.04$ given by \eqref{eq:Cst_classicalKM_uniform} for the continuum limit. For further details on the incoherence-coherence transition of the Kuramoto model, see also Ref.~\onlinecite{Strogatz2000}.

These observations point toward an implementation of numerical methods and measurements as outlined below.

\subsection{Numerical methods}\label{numericalmethods}
We calculate numerical solutions of \eqref{eq: stoch kuramotos model on graphs} with a first-order Euler-Maruyama scheme with a time step $\Delta t = 0.01$.
Initial conditions/phases correspond to low synchrony compliant with incoherence, i.e., either the equidistant state $\theta_k(0):=2\pi k/N$ (Uniform complete graph, Erdös-Rényi graph, regular ring lattice with $r=400$, spherical graph) or the random state where $\theta_k(0)$ (Regular ring lattice with $r=25$, sinuisodal graph, Lorentzian graph) is drawn from the uniform distribution on the interval $[0,2\pi)$ (two types of initial conditions were chosen since other attracting states were present for the regular ring lattice with $r=25$).
To characterize the post-transient dynamics, we use the order parameter
$r(t) = \left|\frac{1}{N}\sum_{j=1}^Ne^{\im\theta_j(t)}\right|$ in Eq.~\eqref{eq:orderparameter}
and measure its temporal minimum and maximum, as well as its time average,
\begin{align}
    \rmin&:=\min_{t\in\mathcal{T}}(r(t)),\\
    \rmax&:=\max_{t\in\mathcal{T}}(r(t)),\\
    \rmean&:=|\mathcal{T}|^{-1}\int_\mathcal{T}r(t)\txtd t,\
\end{align}
where $\mathcal{T}:=[\Tt, T]$ with $\Tt$ being the (estimated) transient time and $T$ the total length of the simulation.

To average over stochastic effects, such as Brownian motion and random graphs (Erdös-Rényi and small-world), we average these measurements over several realizations of solutions of~\eqref{eq: stoch kuramotos model on graphs 1} (i.e., ten realizations to account for Brownian motion for eight (random) graph realizations)
and denote ensemble averages with angular brackets $\langle\cdot \rangle$.
To numerically test Theorem~\ref{prop: cha enrgy methods, incoherence-coherence transition}, we calculate $\rminRZ$, $\rmeanRZ$, and $\rmaxRZ$ for different values of $C$ and compare the resulting curves with $\Cst$. The sampling points for the coupling $C$ are non-uniformly spaced with a higher density in regions of interest (indicated as blue dots in Fig.~\ref{fig:incoh2coh}).

A suitable transient time $\Tt$ can be determined based on the following considerations.
The actual transient is maximal for $C=\Cst$ and decreases for $C>\Cst$; see Fig.~\ref{fig:belowmidaboveCc}.
One could estimate $\Tt$ for each value of $C$ individually to optimize for computational effort; but for simplicity, we estimated the length of $\Tt$ only at $C=\Cst$ and used this $\Tt$ for all probed values of $C$, as this choice guarantees a sufficiently long transient time.
Due to the fluctuations present in the signal of $r(t)$ (pseudo-random fluctuations and stochastic noise),  the estimation of $\Tt$ is heuristic; i.e., it is done by visual inspection. This estimate of $\Tt$ improves with increasing $N$ as the amplitude of (pseudo-random) fluctuations decreases. Taking these considerations into account, we chose $\Tt=700,T=1000,\Delta t=0.01,N=1000$ for all our numerical solutions of \eqref{eq: stoch kuramotos model on graphs 1}.

\begin{figure*}[htp!]
\newcommand{\axx}{2}    
\newcommand{\axyb}{6}
\newcommand{\axyt}{66}
\newcommand{\axy}{1}    
\newcommand{\axxl}{6}
\newcommand{\axxm}{39.5}
\newcommand{\axxr}{88}
\newcommand{\axxlabelposx}{50} 
\newcommand{\axxlabelposy}{-3}
\newcommand{\axylabelposx}{-1} 
\newcommand{\axylabelposy}{22}
\newcommand{\axlabelposx}{-3} 
\newcommand{\axlabelposy}{72}
\newcommand{\plx}{40}
\newcommand{\ply}{72}
\begin{overpic}[width=0.32\linewidth,percent]{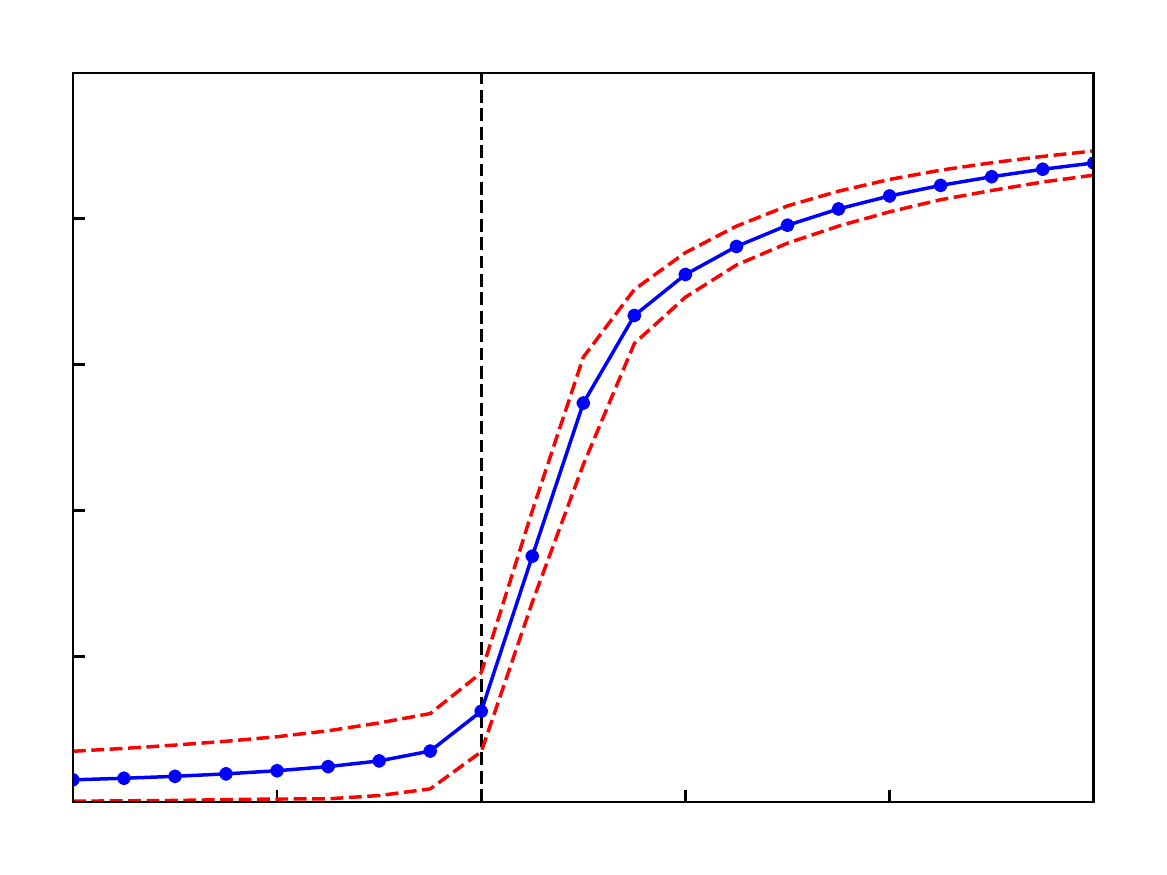}
\put(\axx,\axyt){\color{black}1}
\put(\axx,\axyb){\color{black}0}
\put(\axylabelposx,\axylabelposy){\color{black}\rotatebox{90}{Order parameter}}
\put(\axxl,\axy){\color{black}0}
\put(\axxm,\axy){\color{black}1}
\put(\axxr,\axy){\color{black}2.5}
\put(\axxlabelposx,\axxlabelposy){\color{black}$C/\Cst$}
\put(\axlabelposx,\axlabelposy){(a)}
\put(40,\ply){Uniform complete}
\end{overpic}
\begin{overpic}[width=0.32\linewidth,percent]{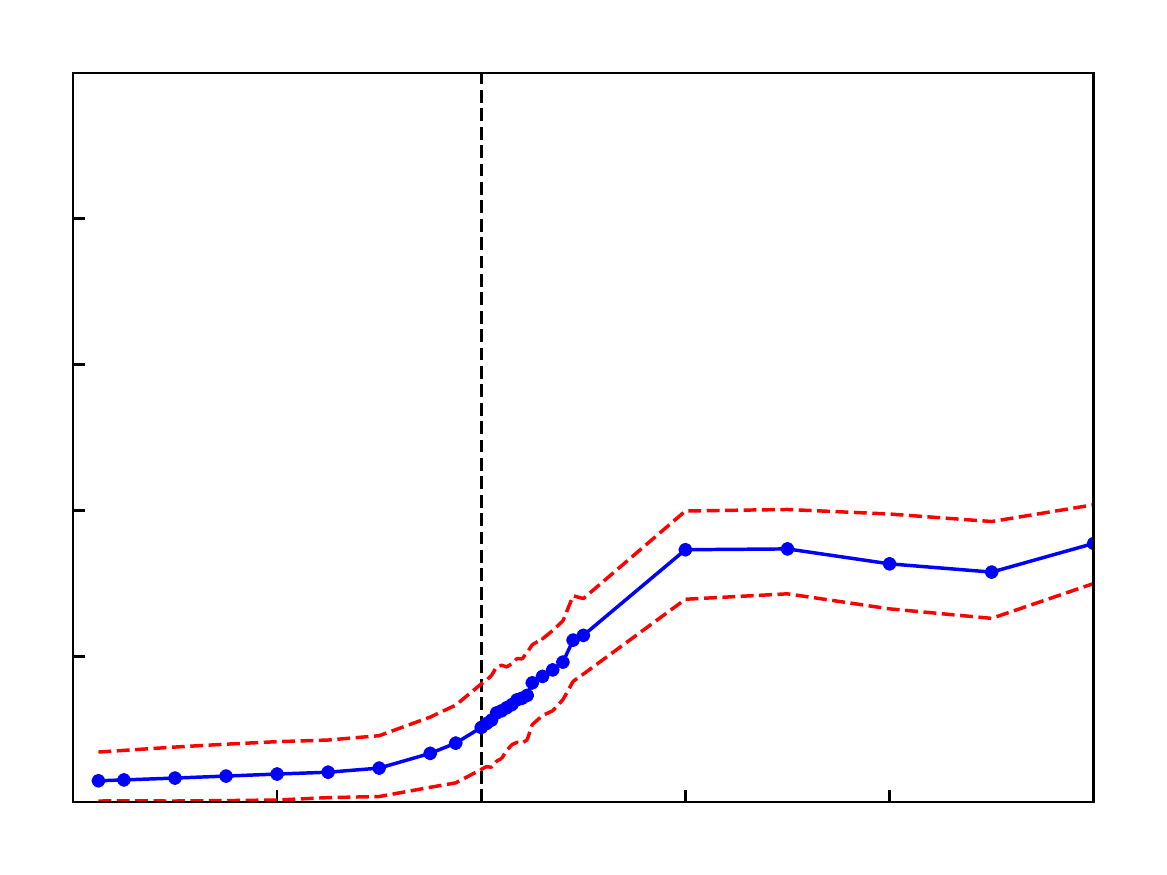}
\put(\axx,\axyt){\color{black}1}
\put(\axx,\axyb){\color{black}0}
\put(\axylabelposx,\axylabelposy){\color{black}\rotatebox{90}{Order parameter}}
\put(\axxl,\axy){\color{black}0}
\put(\axxm,\axy){\color{black}1}
\put(\axxr,\axy){\color{black}2.5}
\put(\axxlabelposx,\axxlabelposy){\color{black}$C/\Cst$}
\put(\axlabelposx,\axlabelposy){(b)}
\put(20,\ply){Regular ring lattice ($r=25$)}
\end{overpic}
\begin{overpic}[width=0.32\linewidth,percent]{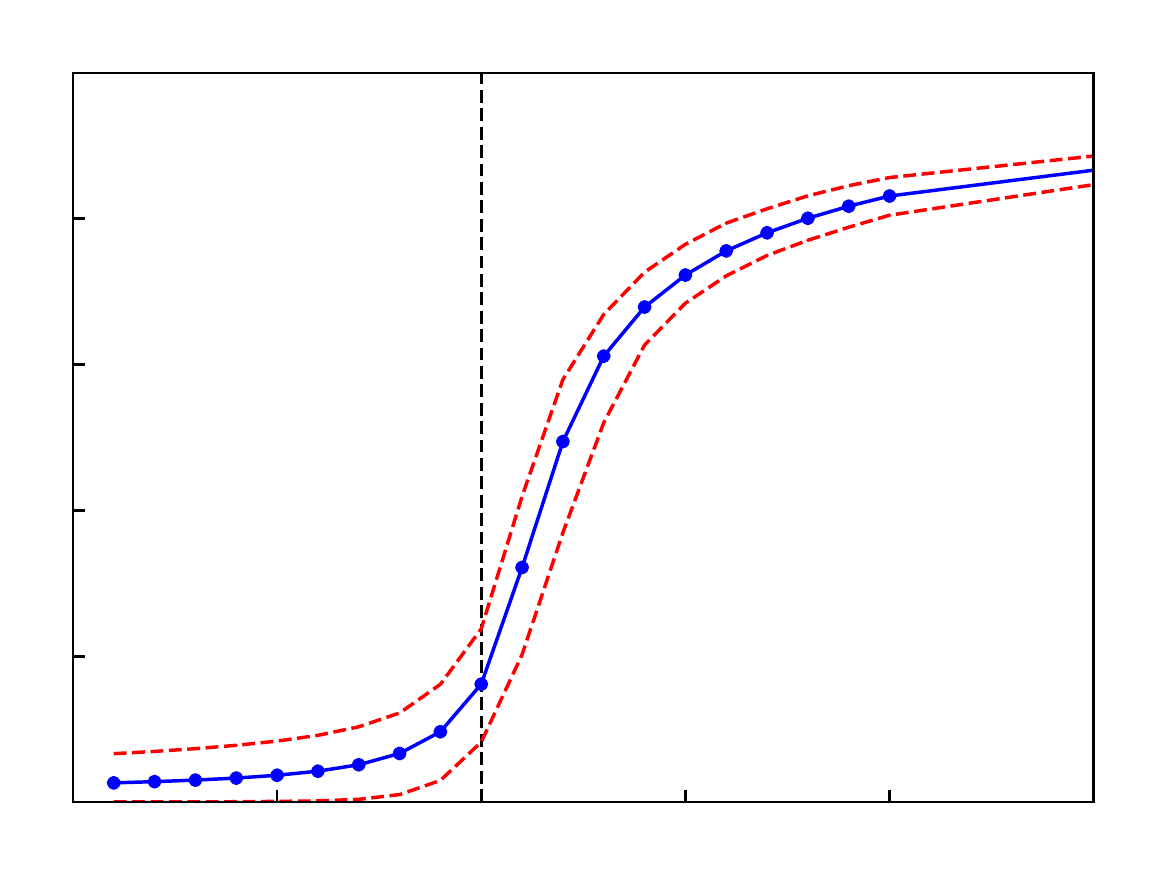}
\put(\axx,\axyt){\color{black}1}
\put(\axx,\axyb){\color{black}0}
\put(\axylabelposx,\axylabelposy){\color{black}\rotatebox{90}{Order parameter}}
\put(\axxl,\axy){\color{black}0}
\put(\axxm,\axy){\color{black}1}
\put(\axxr,\axy){\color{black}2.5}
\put(\axxlabelposx,\axxlabelposy){\color{black}$C/\Cst$}
\put(\axlabelposx,\axlabelposy){(c)}
\put(20,\ply){Regular ring lattice ($r=400$)}
\end{overpic}
\\\medskip\medskip
\begin{overpic}[width=0.32\linewidth,percent]{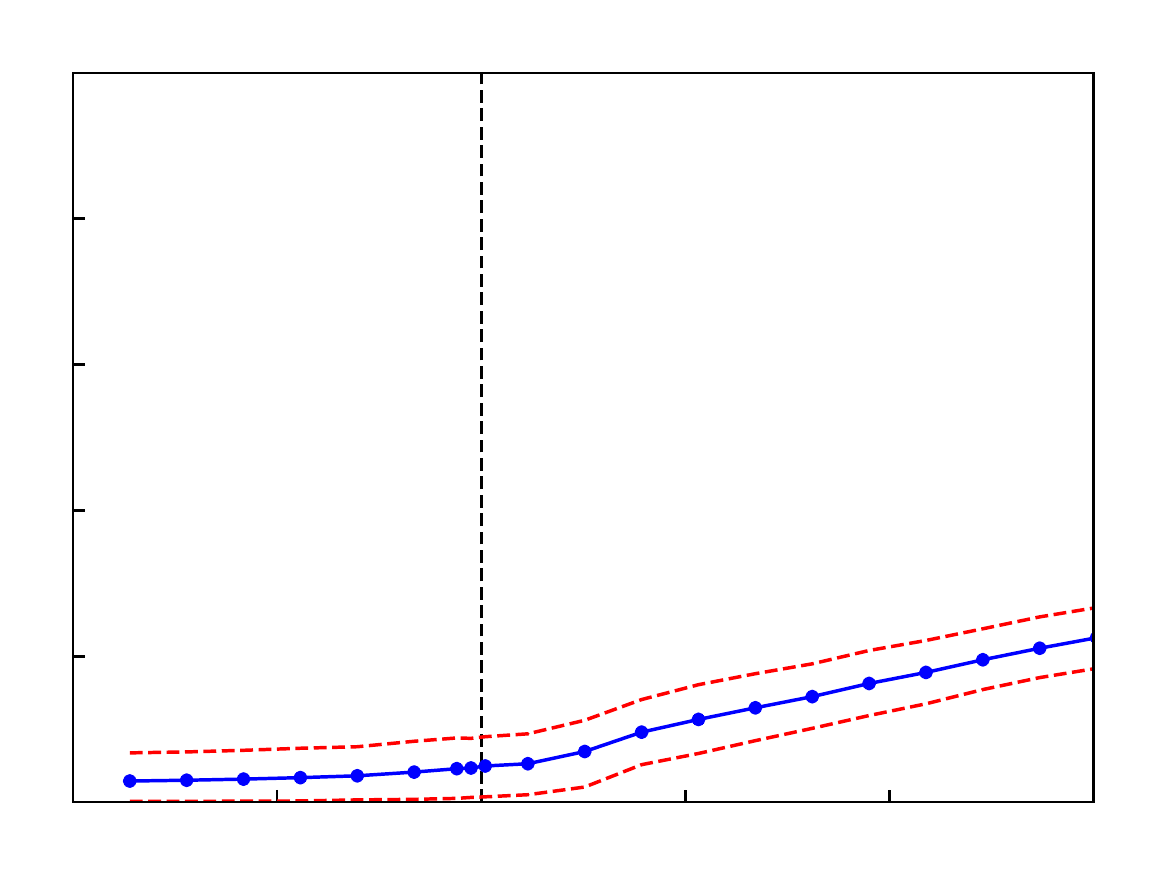}
\put(43,30){\includegraphics[scale=0.24]{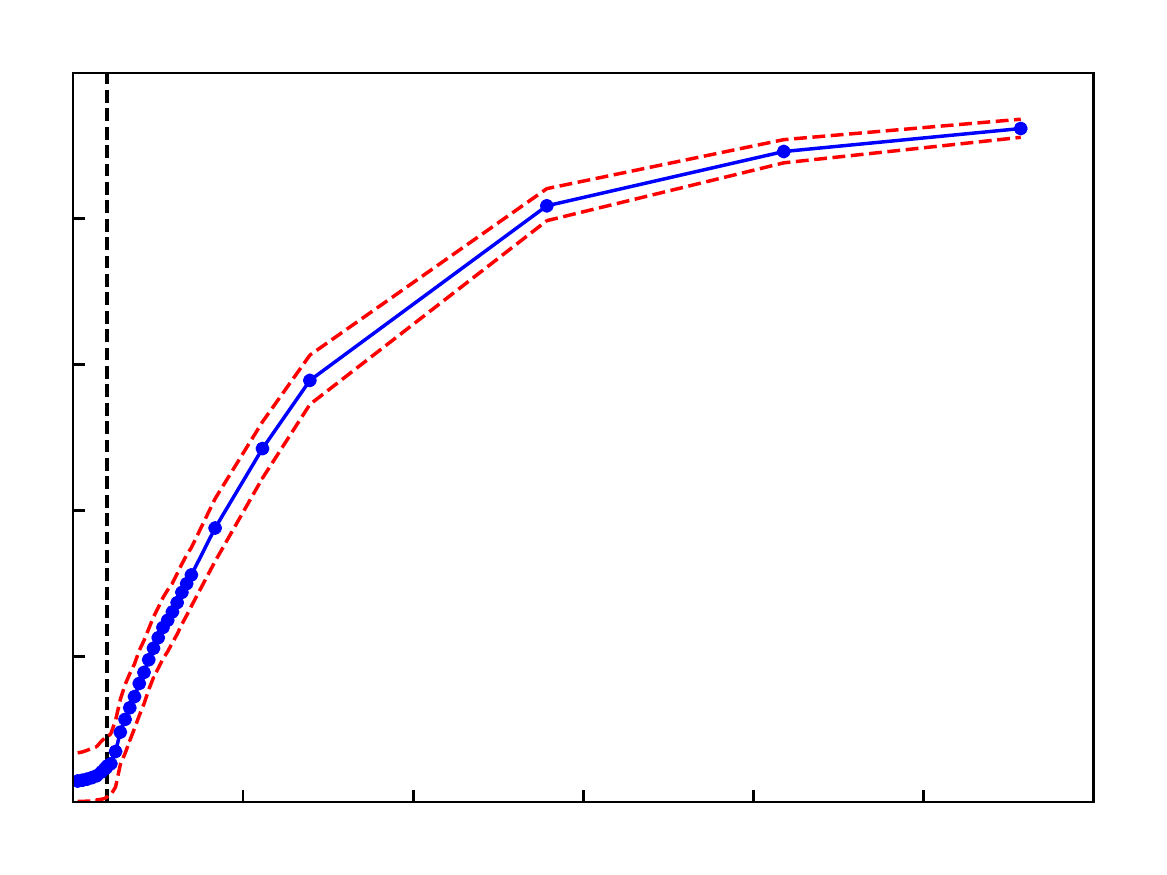}}
\put(46,28){0}
\put(85,28){30}
\put(42,33){0}
\put(42,60){1}
\put(\axx,\axyt){\color{black}1}
\put(\axx,\axyb){\color{black}0}
\put(\axylabelposx,\axylabelposy){\color{black}\rotatebox{90}{Order parameter}}
\put(\axxl,\axy){\color{black}0}
\put(\axxm,\axy){\color{black}1}
\put(\axxr,\axy){\color{black}2.5}
\put(\axxr,\axy){\color{black}2.5}
\put(\axxlabelposx,\axxlabelposy){\color{black}$C/C^{\natural,N}$}
\put(\axlabelposx,\axlabelposy){(d)}
\put(30,\ply){Lorentzian ($\mu=0.01$)}
\end{overpic}
\begin{overpic}[width=0.32\linewidth,percent]{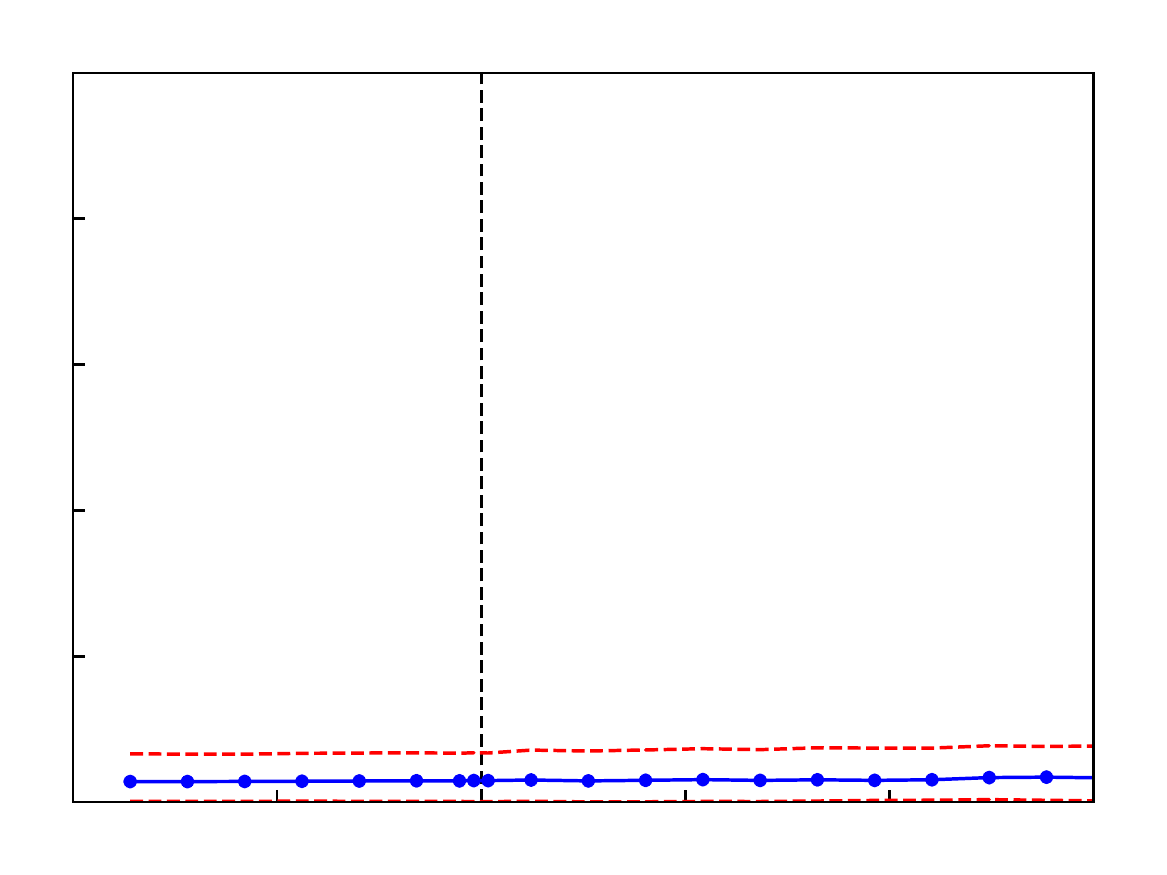}
\put(43,30){\includegraphics[scale=0.24]{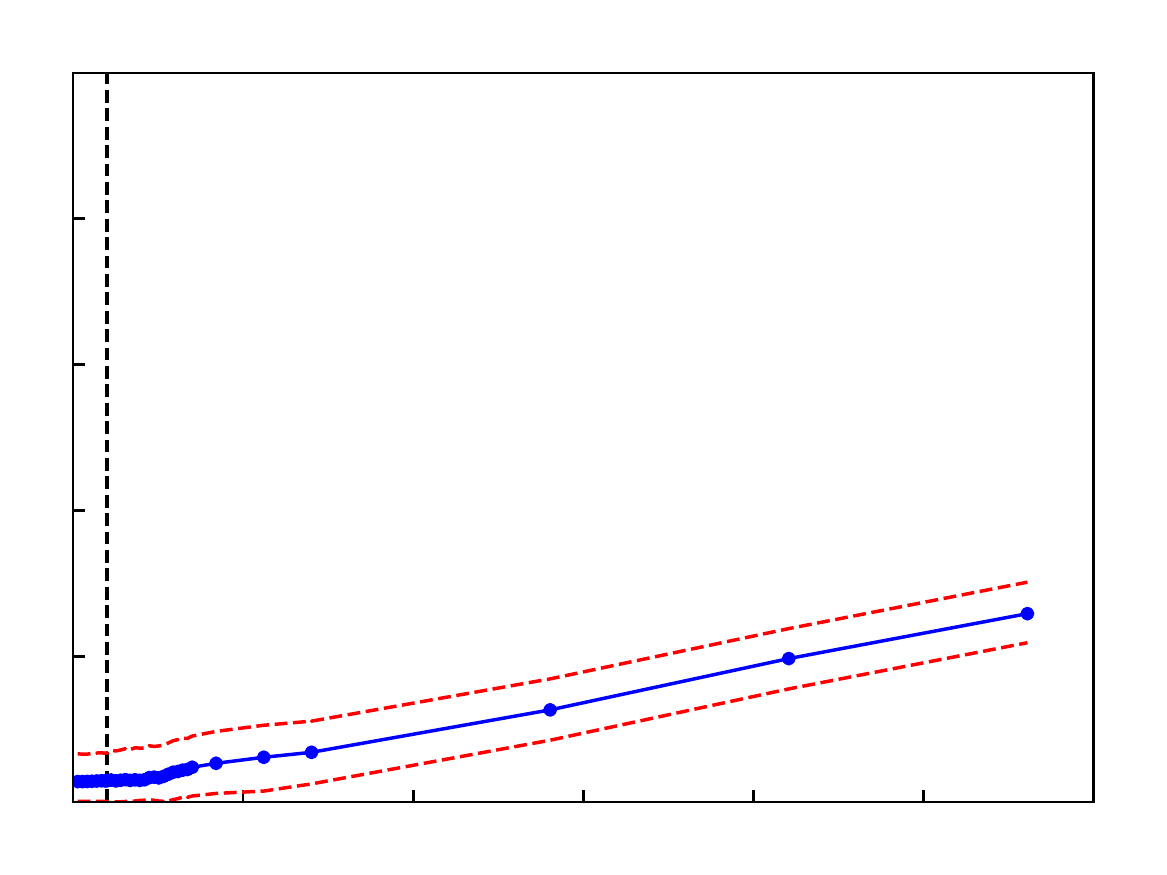}}
\put(46,28){0}
\put(85,28){30}
\put(42,33){0}
\put(42,60){1}
\put(\axx,\axyt){\color{black}1}
\put(\axx,\axyb){\color{black}0}
\put(\axylabelposx,\axylabelposy){\color{black}\rotatebox{90}{Order parameter}}
\put(\axxl,\axy){\color{black}0}
\put(\axxm,\axy){\color{black}1}
\put(\axxr,\axy){\color{black}2.5}
\put(\axxlabelposx,\axxlabelposy){\color{black}$C/C^{\natural,N}$}
\put(\axlabelposx,\axlabelposy){(e)}
\put(30,\ply){Lorentzian ($\mu=0.001$)}
\end{overpic}
\begin{overpic}[width=0.32\linewidth,percent]{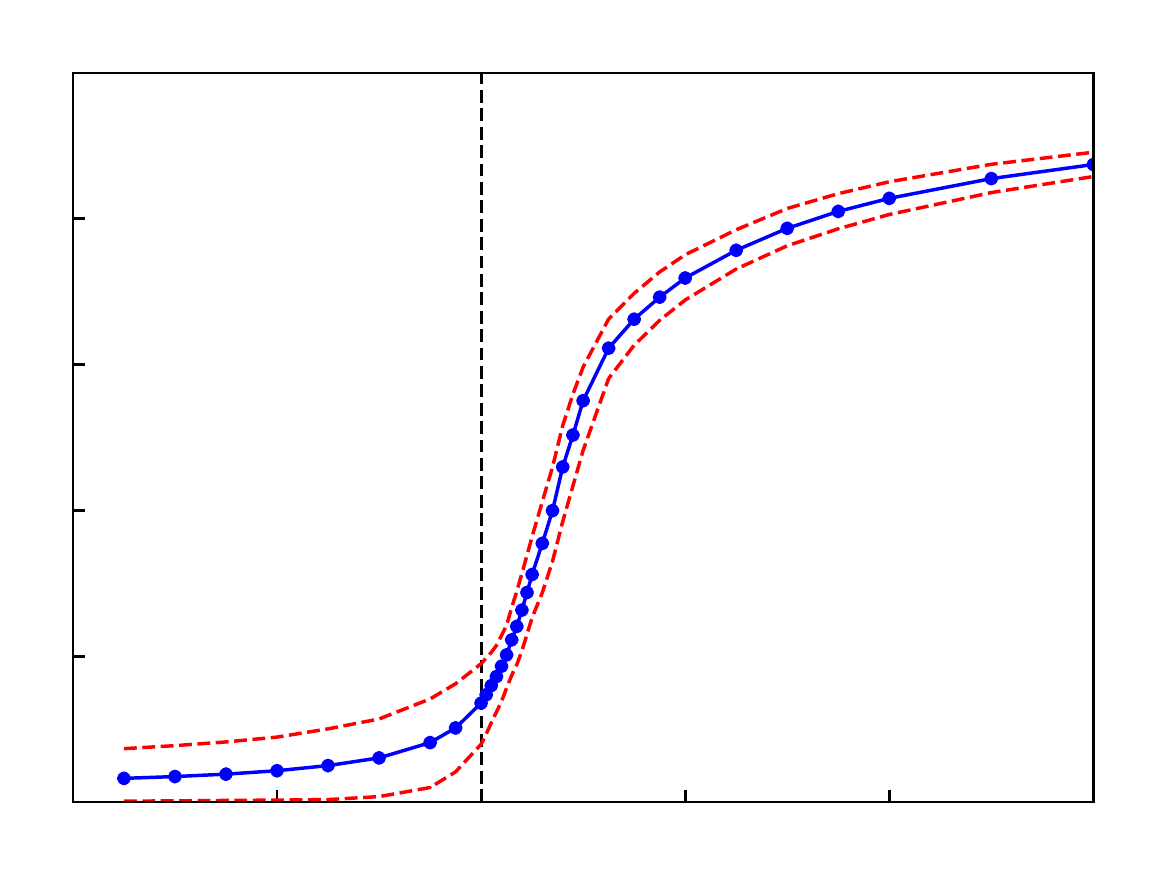}
\put(\axx,\axyt){\color{black}1}
\put(\axx,\axyb){\color{black}0}
\put(\axylabelposx,\axylabelposy){\color{black}\rotatebox{90}{Order parameter}}
\put(\axxl,\axy){\color{black}0}
\put(\axxm,\axy){\color{black}1}
\put(\axxr,\axy){\color{black}2.5}
\put(\axxlabelposx,\axxlabelposy){\color{black}$C/C^{\natural,N}$}
\put(\axlabelposx,\axlabelposy){(f)}
\put(40,\ply){Spherical}
\end{overpic}
\\\medskip\medskip
\begin{overpic}[width=0.32\linewidth,percent]{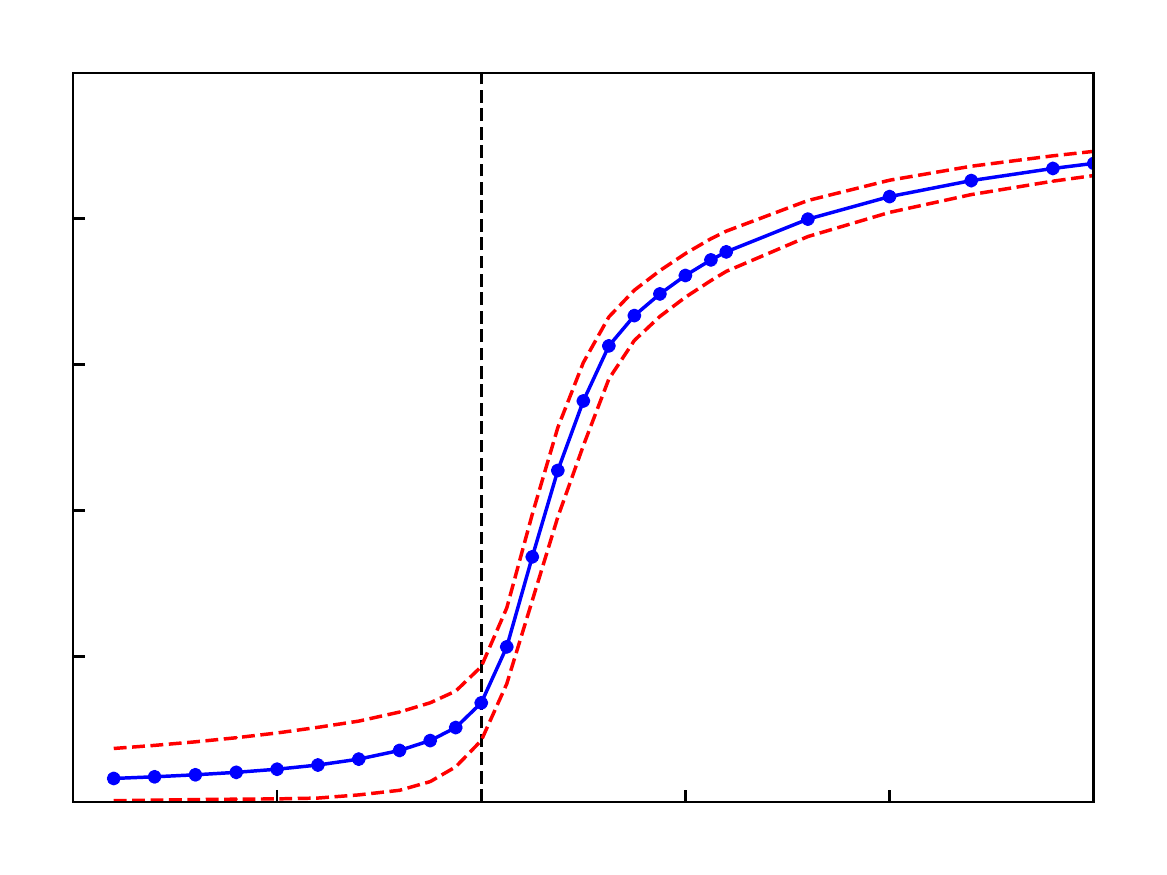}
\put(\axx,\axyt){\color{black}1}
\put(\axx,\axyb){\color{black}0}
\put(\axylabelposx,\axylabelposy){\color{black}\rotatebox{90}{Order parameter}}
\put(\axxl,\axy){\color{black}0}
\put(\axxm,\axy){\color{black}1}
\put(\axxr,\axy){\color{black}2.5}
\put(\axxlabelposx,\axxlabelposy){\color{black}$C/\Cst$}
\put(\axlabelposx,\axlabelposy){(g)}
\put(35,\ply){Small world}
\end{overpic}
\begin{overpic}[width=0.32\linewidth,percent]{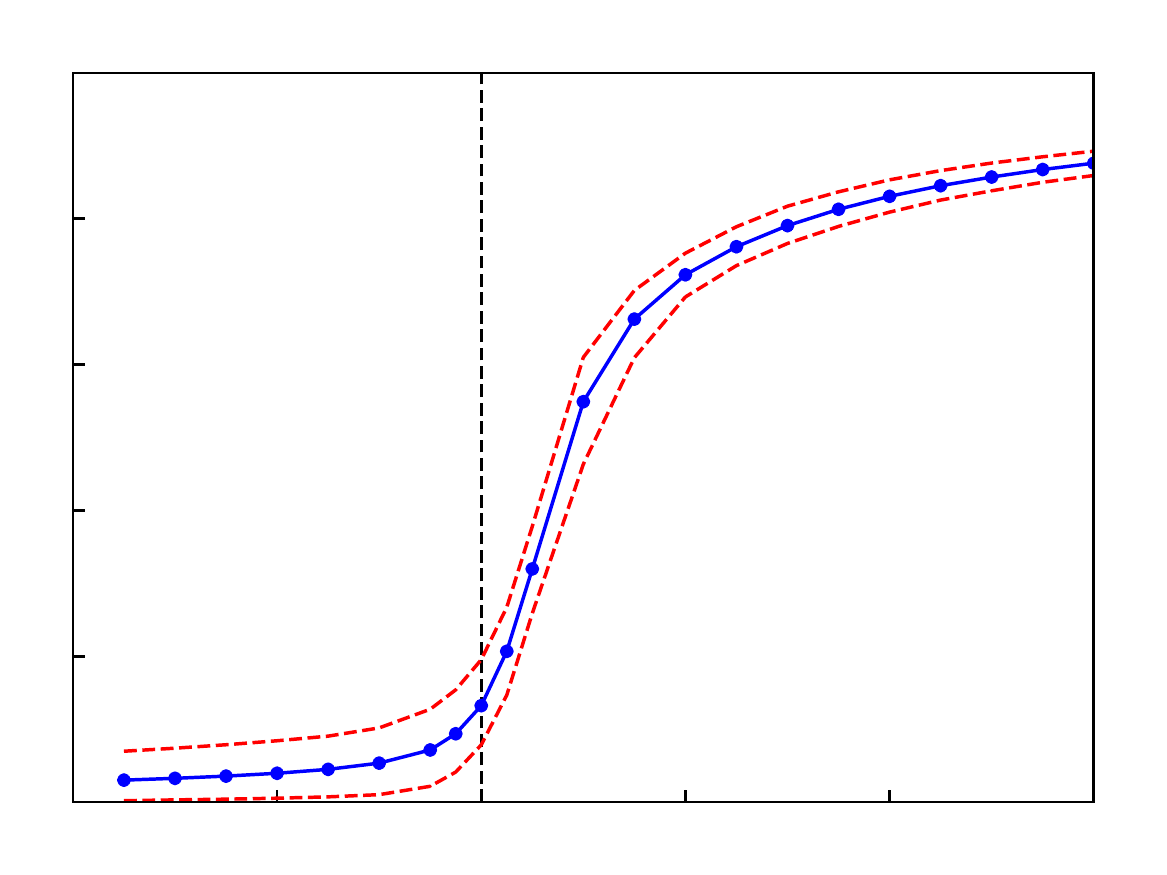}
\put(\axx,\axyt){\color{black}1}
\put(\axx,\axyb){\color{black}0}
\put(\axylabelposx,\axylabelposy){\color{black}\rotatebox{90}{Order parameter}}
\put(\axxl,\axy){\color{black}0}
\put(\axxm,\axy){\color{black}1}
\put(\axxr,\axy){\color{black}2.5}
\put(\axxlabelposx,\axxlabelposy){\color{black}$C/\Cst$}
\put(\axlabelposx,\axlabelposy){(h)}
\put(40,\ply){Sinusoidal}
\end{overpic}
\begin{overpic}[width=0.32\linewidth,percent]{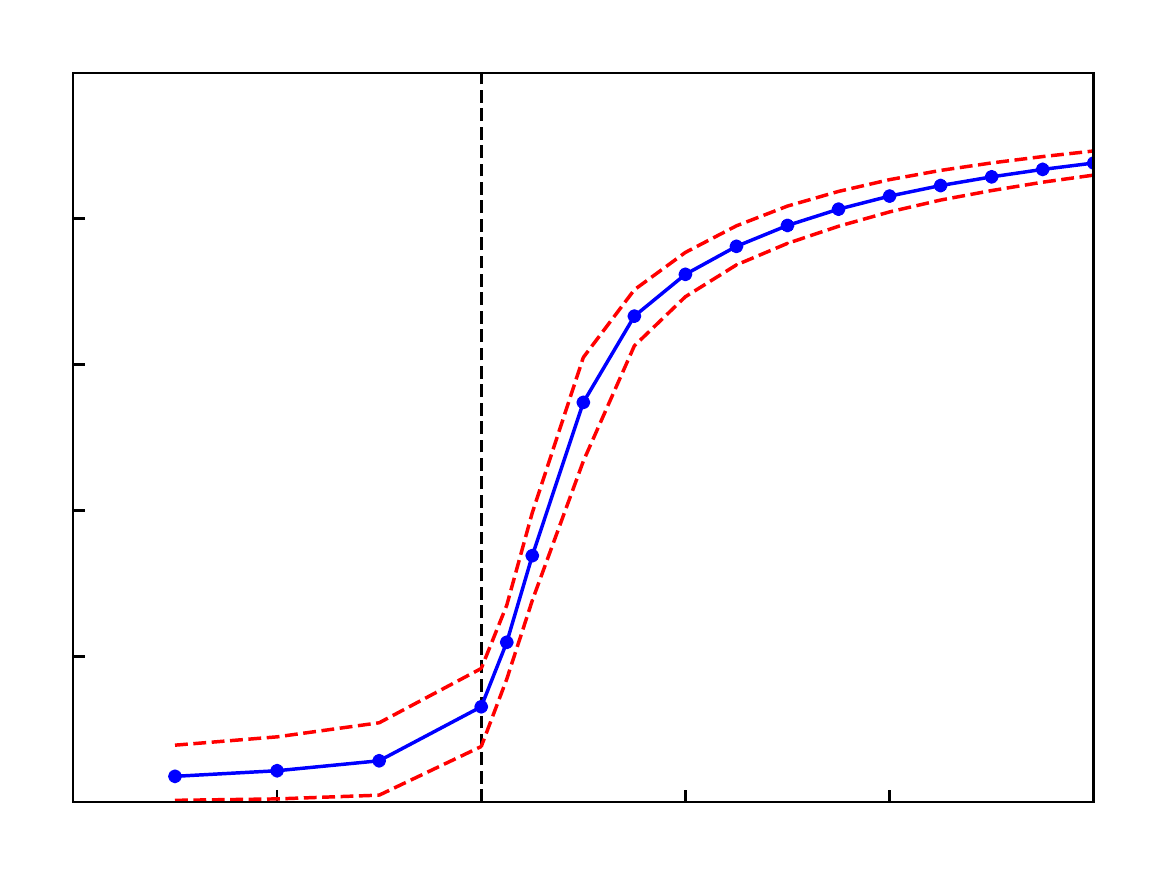}
\put(\axx,\axyt){\color{black}1}
\put(\axx,\axyb){\color{black}0}
\put(\axylabelposx,\axylabelposy){\color{black}\rotatebox{90}{Order parameter}}
\put(\axxl,\axy){\color{black}0}
\put(\axxm,\axy){\color{black}1}
\put(\axxr,\axy){\color{black}2.5}
\put(\axxlabelposx,\axxlabelposy){\color{black}$C/\Cst$}
\put(\axlabelposx,\axlabelposy){(i)}
\put(30,\ply){Erdös-Rényi}
\end{overpic}
\\\medskip
  \caption{
  Incoherence-coherence transition for numerical solutions of \eqref{eq: stoch kuramotos model on graphs 1} for different coupling topologies. The dashed black lines show the respective value of $\Cst$. The red dashed lines show $\rminRZ$ and $\rmaxRZ$, while the blue lines show $\rmeanRZ$.
    a) Uniform complete graph, $\Cst=0.08$.
    b) Regular ring lattice with $r=25,\Cst=0.8$.
    c) Regular ring lattice with $r=400,\Cst=0.05$.
    d) Lorentzian graph with $\mu=0.01,C^{\natural,N}=0.0718$.
    e) Lorentzian graph with $\mu=0.001,C^{\natural,N}=0.0713$.
    f) Spherical graph with $M=50,C^{\natural,N}=0.8003$.
    g) Small-world graph with $k=100,p=0.5,\Cst=0.2$.
    h) Sinusoidal graph, $\Cst=0.08$.
    i) Erdös-Rényi graph with $p=0.5,\Cst=0.08$.
    Parameters for all graphs: $N=1000,\Delta t=0.01,\Tt=700,T=1000,\beta=50,$ ten realizations of Brownian motion, eight realizations of the graph (if random). $C^{\natural,N}$ is evaluated when $\Cst$ cannot be evaluated.
  }\label{fig:incoh2coh}
\end{figure*}

\subsection{Graph(on) topologies and their associated incoherence-coherence transitions}
We now define different graph structures for which we carry out numerical simulations to test for the onset of the incoherence-coherence transition.
Results for the incoherence-coherence transitions for the various graph topologies are summarized in Fig.~\ref{fig:incoh2coh}.
\subsubsection{Incoherence-Coherence threshold for finite and infinite oscillator systems}
We extend our analysis to different coupling topologies while using the coupling interaction  $D(u)=\sin(u)$.
The theoretically predicted threshold for the  incoherence-coherence transition, $\Cst$, valid in the mean-field limit is calculated using \eqref{eq:Cst_classicalKM}.
We shall compare the numerical findings to this theoretical prediction for coupling topologies where it is possible,
However, for certain graphops $A$, a characterization of $\sigma(A)$ exceeds the scope of this study (Spherical graph in Sec.~\ref{sec:spher}; Lorentzian graph in Sec.~\ref{sec:lore}). In such cases, we instead compute the eigenvalues $\sigma^N(A^N)$ of a discrete coupling matrix $A^N$ that approximates $A$.
We expect that the finite-dimensional matrices $A^N$ can be used to provide an approximation to (at least the boundary of) the spectrum of the limiting graphop $A$ as $N\ti$, and therefore, $\Cst$ can be approximated by its discrete corollary
\begin{align}
C^{\natural,N} &= \frac{2}{\beta\Lambda^N(A^N)},
\label{eq:Cst_Nfinite}
\end{align}
where $\Lambda^N(A^N):=\max_{\lambda\in\sigma^N(N^{-1}A^N)}|\lambda|$ is the maximal eigenvalue associated with $A^N$.
Finally, we also mention the possibility of ``spectral pollution,''~\cite{davies2004spectral} which, in principle, can occur when numerically approximating the spectrum of an operator with finite-dimensional matrices. However, as we shall see, our numerical and analytical results are consistent; we, therefore, anticipate that the numerical calculations are sufficiently stable.

\subsubsection{Coupling topologies}\label{sec:topo}
\paragraph{Regular ring lattice with $r$ neighbors.}
Nodes for this coupling topology may be imagined to be arranged on a ring such that every node is linked to a given number of $r$ nearest neighbors.
In the continuum limit $N\ti$ , the ring lattice graphon can be defined as
\begin{align}\label{eq:Kregularringlattice}
K(x,y) = \begin{cases}
1, & \min \{|x-y|, 1 -|x-y|\} \leq h,\\
0 & \text{otherwise},
\end{cases}
\end{align}
where $0\leq h\leq1/2$ is the (continuous) coupling range for oscillators located at $x$ and $y$ on $\Omega$.
The graphop $A$ defined via this graphon kernel $K$ has $\Lambda(A)=2h$ (this can be shown, e.g., by writing $K(x,y)$ as a Fourier series, and the values of $\sigma(A)$ are given in Ref.~\onlinecite{Gao2019}.)

For $N<\infty$, we simply define the regular ring lattice graph via
\begin{align}
A^N_{kj} =  \begin{cases}
1, & k\neq j \text{ and } \min\{|k-j|, 1-|k-j|\}\leq r,\\
0 & \text{else},
\end{cases}
\end{align}
where the (discrete) coupling range $r=r(h)\in[N]$ for oscillators located at $k$ and $j$ in $[N]$ satisfies $0\leq r\leq N/2$ with $N$ even.
It is easy to check that $\Lambda^N(A^N)=2r/N$. In our simulations, we choose $r$ from which the value $h$ for the corresponding graphon kernel follows via $h=r/N$. We then have $\Lambda(A)=\Lambda^N(A^N)=2r/N$.
We note two limiting cases; namely, we obtain all-to-all coupling for $h=1/2$ and zero coupling for $h=0$.

\paragraph{Erdös-Rényi graph.}
The Erdös-Rényi (ER) graph(on) is constructed in a random process where the presence (or absence) of every edge (of the complete graph) is chosen with a probability $p\in[0,1]$.

In the continuum limit, $N\ti$, the Erdös-Rényi graphop simply becomes the complete (all-to-all) graphop with constant uniform coupling strength  $p$; i.e., the corresponding graphon kernel is  $K(x,y)= p$; see Ref.~\onlinecite{Medvedev2014c}.
It follows then that $\Cst=2/(\beta p)$.

For finite oscillators $N<\infty$, a realization of the ER graph on $N$ nodes may be obtained by drawing $N(N-1)/2$ random numbers $a_{kj},\,\, 1\leq k < j\leq N$, from the uniform distribution on the interval $[0, 1]$. The adjacency matrix of the graph is then
\begin{align}
A_{kj}^N = A_{jk}^N = \begin{cases}
1, & k<j,\,\,\,a_{kj}\leq p,\\
0 & \text{else}.
\end{cases}
\end{align}
For $p=1$, we obtain all-to-all coupling with uniform coupling strength 1 (complete graph), while $p=0$ yields zero coupling.

\paragraph{Small-world graph.}
The small world (SW) graph~\cite{WattsStrogatz1998}
interpolates between a regular ring lattice and a ER graph structure,
thus creating a topology that is quite regular but also entertains random links across the network. This structure results in short path lengths even when nodes are far away on the ring.

For finite graphs, $N\ti$, the small world graphop $A$ can be constructed via the graphon kernel~\cite{Medvedev2014a,Medvedev2014b} given by
\begin{align}
    K(x,y) &= (1-p)W(x,y) + 2ph,\
\end{align}
where
\begin{align}
W(x,y) &= \begin{cases}
1, & \min \{|x-y|, 1 -|x-y|\} \leq h,\\
0 & \text{else},
\end{cases}
\end{align}
with (continuous) coupling range $0\leq h\leq1/2$ (note that $W(x,y)$ is identical to $K(x,y)$ in \eqref{eq:Kregularringlattice} further above for the regular ring lattice). It can be shown that $\Lambda(A)=2h$ (to see this, one needs to write  $K(x,y)$ in terms of a Fourier series; the values of $\sigma(A)$ are given by Gao and Caines~\cite{Gao2019}).

For $N<\infty$, realizations of the SW graph on $N$ nodes may be obtained via the procedure introduced by Watts and Strogatz~\cite{WattsStrogatz1998}: One starts with a regular ring lattice on $N$ nodes with $r$ nearest neighbors (discrete coupling range). One selects a  constant probability $p\in[0,1]$. For each node $k$ and each link between $k$ and its $r$ nearest neighbors to the right, we draw a random number $X\in[0,1]$ i.i.d. from the uniform distribution. If $X\leq p$, we draw a random integer $j$ from the uniform distribution on $[N]$. If $k\neq j$ and the edge $(k,j)$ does not yet exist, it is created and the old link deleted.

In our numerical setting, we simply pick a value $r$ and the value $h$ for the corresponding graphon kernel follows from $h=r/N$. We numerically confirmed that $\Lambda^N(A^N)\approx\Lambda(A)=2h$. We shall thus use the value $C^{\natural,N}\approx\Cst=1/(\beta h)$.

\paragraph{Spherical graph.}\label{sec:spher}
The action of the spherical graphop $A:L^2(\mathbb{S}_2)\mapsto L^2(\mathbb{S}_2)$ on a function $f$ is defined by
\begin{align}
(Af)(x) &= \int_{y\perp x}f(y)\txtd\nu_x,\label{eq: def sphergrp}
\end{align}
where $\nu_x$ is the uniform measure.
The spherical graphop thus integrates $f$ over the circle on the unit sphere that consists of all the points perpendicular to $x$. The resulting circle is the \textit{equator} of the point $x$.
The spherical graphop does not have a graphon kernel or a known spectrum; therefore, we need to calculate $\Cst$ via \eqref{eq:Cst_Nfinite}.
Moreover, a matrix approximation to the spherical graphop has to our knowledge not yet been proposed. Here, we propose a possible approximation without claiming any convergence properties as $N\ti$.
Choosing $N$ (approximately equidistant) sample points $x_1,\ldots,x_N$ on the unit sphere, we may obtain a matrix $A^N$ approximating $A$ by defining $A^N_{kj}=A^N_{jk}=1$ if $x_k$ and $x_j$ are approximately perpendicular; otherwise, $A^N_{kj}=A^N_{jk}=0$.
The discretized version of \eqref{eq: def sphergrp} then reads
\begin{align}
    (A^Nf)_k = \frac{1}{N}\sum_{j=1}^NA^N_{kj}f(x_{j}).
\end{align}
We refer to $A^N$ as a \textit{spherical graph}.
Three requirements should be made on $A^N$. For each point $x_k$, the points $x_j$ for which $A^N_{kj}=1$ should (i) lie sufficiently close to the equator of $x_k$, (ii) be sufficiently equidistant, and (iii) be (almost) equally many for all $k$.
Clearly, if we take an arbitrary point on the unit sphere, one can place $M$ perfectly equidistant points on its equator.
However, (i) and (ii) must be fulfilled reasonably well for \textit{all} $N$ points and their respective equators. Therefore, the points should form a regular grid. While a perfectly regular grid of $N>6$ points on the sphere is impossible, there exist approximately regular grids~\cite{Kogan2017}. Here, we place the points in a spiral of width $0.1+1.2N$ around the sphere, starting and ending (approximately) on the poles. This method is implemented in the Mathematica Software package~\cite{reference_wolfram_2021_spherepoints}.
We denote the set of points with this spacing on the unit sphere as $P$.
The task is to determine subsets $E_k\subset P, 1\leq k\leq N$, such that each $E_k$ discretizes the equator of $x_k$. To this end, we first calculate $p_{kj}:=|\langle x_k,x_j \rangle|$ for each $1\leq k<j\leq N$, to determine how close the pairs of points are to being perpendicular.
Then, we specify $M$, the desired (approximate) cardinality of all $E_k$'s.
Now, we can, for each $k$, find the (approximately) $M$ points $x_j$ with the smallest values of $p_{kj}$ and make them members of $E_k$,
under the constraint that if $x_k\in E_j$, then $x_j\in E_k$, to ensure that $A^N$ is symmetric.
We end up with a (symmetric) $A^N$ that fulfills demands (i) and (ii) in an acceptable manner, while demand (iii) is fulfilled well: $|E_k|$ is either $M$ or $M-1$ for all $1\leq k\leq N$ (see Fig.~\ref{sphergrp}).
\begin{figure}
\centering
\newcommand{\axxtwo}{14}
\newcommand{\axytwo}{2}
  \begin{overpic}[width=\columnwidth,percent]{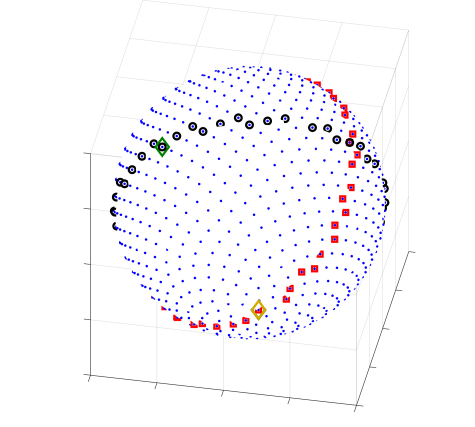}
    \put(16,6){\color{black}-1}
    \put(44.5,2){\color{black}0}
    \put(73,-1){\color{black}1}
    \put(40,0){\color{black}$x$}
    \put(\axxtwo,56){\color{black}1}
    \put(\axxtwo,33.5){\color{black}0}
    \put(\axxtwo,10){\color{black}-1}
    \put(10,37){\color{black}$z$}
    \put(93,36){\color{black}1}
    \put(88,23){\color{black}0}
    \put(81,9){\color{black}-1}
    \put(94,19){\color{black}$y$}
  \end{overpic}

  \caption{
  Matrix approximation for the spherical graphop. Blue dots indicate the sample points, black circles mark points belonging to the discretized equator of one of the exemplary points (yellow diamond), and red squares mark points belonging to the discretized equator of the other exemplary point (green diamond). The two exemplary points are thus members of each other's discretized equators. Parameters are  $N=1000$ and $M=50$.
  }
  \label{sphergrp}
\end{figure}

Since the spherical graphop cannot be defined via a graphon, we find $\Lambda^N(A^N)=0.04998$, and thus, $C^{\natural,N}=0.8003$.

\paragraph{Sinusoidal graph.}
In the sinusoidal coupling topology, nodes are coupled most strongly to their nearest neighbors, the coupling then smoothly decreases the farther neighbors are apart, finally the coupling is zero between nodes opposite on the ring.
We define the graphon kernel as
\begin{align}
K(x,y) = \frac{1}{2}(1+\cos2\pi(x-y)).
\end{align}
It can be shown~\cite{Gao2019} that the graphop $A$ induced by $K$ has $\Lambda(A)=1/2$ so that $\Cst = 4/\beta = 0.08$.

For $N<\infty$, we define the matrix $A^N$ by
\begin{align}
A^N_{k,j} = \frac{1}{2}\left(1+\cos{\left(2\pi\frac{k-j}{N}\right)}\right).
\end{align}

\paragraph{``Lorentzian'' graph.}\label{sec:lore}
We also consider graphs for which a mean-field description is more challenging, and which, therefore, could potentially fail to exhibit the behavior predicted by Theorem~\ref{prop: cha enrgy methods, incoherence-coherence transition}. A good candidate would be an irregular and sparse graph with few very strong links, while the vast majority of links is very weak. We can define such a topology based on the Lorentzian (graphon) kernel,
\begin{align}
\begin{split}
K(x,y)&= \frac{\mu/\pi}{(x-x_0)^2+(y-y_0)^2+\mu^2} \\&+ \frac{\mu/\pi}{(x-y_0)^2+(y-x_0)^2+\mu^2},
\end{split}
\end{align}
where $x_0,y_0\in[0,1]$. $K(x,y)$ peaks in the points $(x_0,y_0)$ and $(y_0,x_0)$, and converges to a sum of delta distributions centered at these points as $\mu\tz$.
We approximate this graphon in the finite representation as
\begin{align}
\begin{split}
A^N_{kj}&=\frac{\mu/\pi}{\left(\frac{k}{N}-x_0\right)^2+\left(\frac{j}{N}-y_0\right)^2+\mu^2}
\\
&
+
\frac{\mu/\pi}{\left(\frac{k}{N}-y_0\right)^2+\left(\frac{j}{N}-x_0\right)^2+\mu^2}.
\end{split}
\end{align}
We use the values $x_0=0.25,y_0=0.75$ with $\mu=0.01$ or $\mu=0.001$.
Computing the spectrum $\sigma(A)$ of the graphop $A$ defined by $K$ exceeds the scope of this work, and we use the eigenvalues of $A^N$,
$$
\Lambda^N(A^N)\evat{\mu=0.01}=0.5573,\;\;\;\Lambda^N(A^N)\evat{\mu=0.001}=0.5612,
$$
to obtain
$$
C^{\natural,N}\evat{\mu=0.01}=0.0718,\;\;\;C^{\natural,N}\evat{\mu=0.001}=0.0713.
$$


\section{Conclusion and Outlook}
\label{sec:discuss}
We formulated a mean-field theory for stochastic phase oscillator models with nontrivial coupling, i.e., heterogenous graph topologies and coupling weights. Our analysis for Kuramoto-type models with odd symmetric coupling functions, obtained via linearization around the incoherent solution branch, yields an exact formula for the critical coupling strength $ \Cst $ at the incoherence-coherence transition in the mean-field limit. Numerically integrating finite representations (see Eq.~\eqref{eq: stoch kuramotos model on graphs}) agrees very well with the predicted threshold $\Cst$ (Eq.~\eqref{eq:Cst_classicalKM}) for a wide range of heterogeneous graph structures (see Fig.~\ref{fig:incoh2coh})\footnote{Note that the regular ring lattice  with $N=1000$  displays imperfect synchronization ($0 \ll r \ll 1 $ ) for $r=25$ (Fig.~\ref{fig:incoh2coh} panel (b)), while $r=400$ a more regular emergence of (partial) coherence (Fig.~\ref{fig:incoh2coh}panel (c)); indeed, this case (with zero noise) is known to exhibit multistability between the coherent branch and a so-called twisted state as long as $r/N<0.34$~\cite{Wiley2006}.}.
We, therefore, expect our theory to be applicable to a large range of applications with heterogeneous oscillator interactions, such as systems with non-uniform coupling associated with chimera states~\cite{PanaggioAbrams2014} or $XY$-oscillators type models with random coupling~\cite{SherringtonKirkpatrick1975,HongMartens2022}.

For certain graph topologies characterized by strong sparsity, large variance in coupling strengths, or other types of ``clusterization'' implying  coupling fragmentation in the network, the mean-field description is expected to break down, in particular, also in terms of correctly predicting the incoherence-coherence transition for finite-size systems. We found that such a problem occurs at least for one instance, namely, for the Lorentzian graph topology (see Fig.~\ref{fig:incoh2coh} (d) and ~\ref{fig:incoh2coh}(e)), for which the detection of a sharp transition point numerically is difficult.
The Lorentzian graph is characterized by only a few nodes with very strong edge weights, while the vast majority of edge weights are very weak: the graph topology is effectively very sparse. This implies that we need a much larger $C$ to observe coherent oscillations.
As becomes apparent from comparing panels (d) and (e), the different quality of the incoherence-coherence transition between the Lorentzian and the other graphs considered is especially pronounced as the effective sparsity increases ($\mu\rightarrow 0$).
Note that not merely larger overall coupling strength $C$ is needed to achieve (partial) coherence, when compared to other topologies; if that were the case, one would just observe larger $\Cst$ for the Lorentzian graph as compared to the other graphs, and the coherence onset would still set in at $C=\Cst$. Rather, the onset of coherence appears to be delayed beyond $C=\Cst$ so that the increase of partial synchrony sets in very slowly as $C$ increases. This observation becomes especially pronounced for very small $\mu$ so that the coupling kernel becomes effectively very sparse. Thus, the Lorentzian graph represents an interesting coupling topology that demarcates a possible class of graphs for which --- at least for certain values of $\mu$ --- our mean-field description and prediction for the incoherence-coherence transition for the finite-size representation break down.

While we extended the mean-field theory for the stochastic Kuramoto model with all-to-all connectivity and uniform coupling strengths to heterogeneous connectivity  with non-uniform coupling strengths, certain constraints apply to our model.  These may limit the validity of our theory and prompt avenues for future research.
For instance, we have assumed that the coupling function $D(u)$ is odd. This assumption excludes, in particular, the Kuramoto-Sakaguchi model, which has a coupling function $D(u)=\sin{(u+\alpha)}$ with a phase-lag $\alpha$. This phase-lag allows one to tune the coupling interaction to be a sine function vs cosine, distinguishing gradient-like and integrable dynamics, respectively (compare with Eqs.~\eqref{eq: stoch kuramotos model on graphs 1} without noise ($\beta^{-1}=0$) and implies different incoherence-coherence transitions (Note that a mix of such interaction is also  essential to observe symmetry breaking chimera states with nonuniform synchronization patterns on the network~\cite{PanaggioAbrams2014,BurylkoMartensBick2022}) --- extending our theory to such interactions would be of interest.
Coupling functions $D(u)$ of higher harmonic order have recently attracted much interest, which imply more complicated stability regimes and transitions between incoherence and coherence~\cite{AshwinBurylko2015,Ashwin2016,Kuramoto2003,BickBoehleKuehn}. Moreover, interactions with arbitrary (e.g. non-symmetric) coupling interactions $D=D(u_k,u_l)$ are possible~\cite{PietrasDaffertshofer2019} which imply directed graph topologies~\cite{KuehnXu}. While we studied the Kuramoto model with identical intrinsic frequencies, the presence of distributed frequencies is also of interest. Finally, extensions to other phase oscillator models, such as the Kuramoto model with inertia~\cite{Ermentrout1991,Rohden2012} or the theta neuron (or QIF neuron) that only performs rigid rotations corresponding to spiking above a threshold current, are worth mentioning.
It would be very useful to derive rigorous mean-field descriptions for the above-mentioned systems; today, mean-field descriptions are available only for full graph structures~\cite{MontbrioPazoRoxin2015,Kuehn2019power}.
Finally, one might also consider transitions between --- or bifurcations of --- states other than incoherence or coherence, such as chimera states or twisted states. Twisted states arise in bifurcations due to negative eigenvalues from the graph operator~\cite{chiba2018bifurcations}. It would be interesting to extend the mean-field theory developed here to such cases. Some work in these directions has been done in the context hypergraphs~\cite{kuehn2022vlasov}.

Another important avenue for future research is to clarify the validity regime for mean-field descriptions for very sparse and very heterogeneous structures. Note carefully that the effective dimension of the VFPE mean-field equation~\eqref{eq: VFPE3} will grow the more heterogeneous the graph is due to the dependence of the node type encoded by points in $\Omega$. Hence, a mean-field description can still exist, and our results indicate that this mean-field is often still very useful to determine whether some number of nodes starts to transition from incoherence to partial synchronization. Yet, for more complex patterns, involving an interplay between all different mean-field node types on very sparse structures, we anticipate that the mean-field description will eventually not be of much use as it is also high-dimensional. In summary, to fully determine the theoretical and practical limitations of heterogeneous mean-field VFPEs remains a challenging problem for future work.\medskip

\section{Acknowledgments}
MAG and CK gratefully thank the TUM International Graduate School of Science and Engineering (IGSSE) for support via the project ``Synchronization in Co-Evolutionary Network Dynamics (SEND).'' BJ and EAM acknowledge the DTU International Graduate School for support via the EU-COFUND project ``Synchronization in Co-Evolutionary Network Dynamics (SEND)''. CK also acknowledges partial support by a Lichtenberg Professorship funded by the Volkswagen Stiftung.

\section{Data Availability Statement}
Data sharing is not applicable to this article as no new data were created or analyzed in this study.


\bibliographystyle{unsrt}

\end{document}